%% file: well-founded-lmcs.tex
\PassOptionsToPackage{final}{graphicx}
\PassOptionsToPackage{final}{hyperref}
\PassOptionsToPackage{notref,notcite}{showkeys}
\documentclass[final]{lmcs}

\numberwithin{equation}{section}

\pdfoutput=1 %

\usepackage{ifthen}

\newboolean{spellingmode}
\setboolean{spellingmode}{false}

\newboolean{proofsinappendix}
\setboolean{proofsinappendix}{true}

\usepackage[footnote,marginclue,nomargin,final]{fixme}
\fxuselayouts{inline,nomargin}
\fxusetheme{color}
\usepackage{xspace}
\usepackage{fdsymbol}
\usepackage{tikz}
\usepackage{ifdraft}
\usepackage{booktabs}
\usepackage{fontawesome}
\usepackage{environ}
\usepackage{newfile}
\usepackage{hyperref}
\usepackage{showkeys}
\usepackage{microtype}

\ifthenelse{\boolean{spellingmode}}{%
	\nolinenumbers
	\AtEndPreamble{
	\hyphenpenalty=50000 %
	\pagestyle{empty}%
	\thispagestyle{empty}%
	\def\pagestyle#1{}%
	\def\thispagestyle#1{}%
	\def\labelmarginpar#1{}%
	}
}{}

\ifdraft{
\overfullrule=4pt %
}{
\overfullrule=0pt %
}

\input{preamble.tex}

\usetikzlibrary{cd,calc}

\def\resettheorembrackets{
\def\theorembracketopen{(}
\def\theorembracketclose{)}
}
\resettheorembrackets
\makeatletter
\def\@spopargbegintheorem#1#2#3#4#5{\trivlist
      \item[\hskip\labelsep{#4#1\ #2}]{#4{\theorembracketopen}#3{\theorembracketclose}\@thmcounterend\ }#5}
\makeatother

\newcommand{\resetCurThmBraces}{%
  \gdef\curThmBraceOpen{(}%
  \gdef\curThmBraceClose{)}}
\resetCurThmBraces
\newcommand{\removeThmBraces}{%
  \gdef\curThmBraceOpen{}%
  \gdef\curThmBraceClose{}}
\resetCurThmBraces

\newenvironment{notheorembrackets}{\removeThmBraces}{\resetCurThmBraces}

\makeatletter
\ifthenelse{\boolean{proofsinappendix}}{
  \newoutputstream{proofstream}
  \openoutputfile{\jobname-proofs.out}{proofstream}
}{
}

\newcommand{\proofappendixbegin}[2]{%
  \phantomsection%
  \subsection*{\textbf{#1~\autoref{#2}}}%
  \addcontentsline{toc}{subsection}{#1~\autoref{#2}}%
  \label{#2:proof}%
  \def\proofappendix@qedsymbolmissing{\qed}
}
\newcommand{\proofappendixend}{%
  \proofappendix@qedsymbolmissing%
}
\let\oldqedhere\qedhere
\def\qedhere{\global\def\proofappendix@qedsymbolmissing{}\oldqedhere}

\newenvironment{proofhere}[2][Proof of]{%
  \proof{#1~\autoref{#2}}%
}{%
  \endproof%
  \par%
}

\ifthenelse{\boolean{proofsinappendix}}{%
  \NewEnviron{proofappendix}[2][Proof of]{%
    \addtostream{proofstream}{\noexpand\proofappendixbegin{#1}{#2}}%
    \addtostream{proofstream}{\noexpand\takeout{}}
    \expandafter\addtostream{proofstream}{\expandonce{\BODY}}%
    \addtostream{proofstream}{\noexpand\proofappendixend}%
    \addtostream{proofstream}{}%
  }%
}{%

}
\makeatother

\tikzstyle{shiftarr}=[
        rounded corners,%
        to path={--([#1]\tikztostart.center)
                     -- ([#1]\tikztotarget.center) \tikztonodes
                     -- (\tikztotarget)},
]

\tikzset{
    commutative diagrams/.cd,
    arrow style=tikz,
    diagrams={>={Straight Barb[length=1.75pt,width=3.85pt,inset=1.95pt]}}, %
    row sep=large,
    column sep = huge
}

\tikzset{cong/.style={draw=none,edge node={node [sloped, allow upside down, auto=false]{$\cong$}}},
         iso/.style={draw=none,every to/.append style={edge node={node [sloped, allow upside down, auto=false]{$\cong$}}}}}

\usetikzlibrary{decorations.pathmorphing}

\tikzcdset{scale cd/.style={every label/.append style={scale=#1},
    cells={nodes={scale=#1}}}}

\newcommand{\set}[2][]{%
  \ifthenelse{\equal{#2}{}}{%
    \ensuremath{{#1\emptyset}}%
  }{%
    \ensuremath{{#1\{#2#1\}}}%
  }%
}

\newcommand{\Z}{\ensuremath{\mathbb{Z}}}

\newcommand{\injmono}{\ensuremath{\overset{\smash{\raisebox{-1pt}{\ensuremath{\scriptstyle+}}}}{\rightarrowtail}}}%
\newcommand{\hpresint}{\ensuremath{H\cap}\xspace}

\newcommand{\agdaurl}{https://arxiv.org/src/2507.18539v4/anc/html/}
\newcommand{\agdafileref}[1]{%
  \href{\agdaurl#1.html}{\texttt{#1}{\tiny\,\faExternalLink}}
}

\usepackage{etoolbox}
\patchcmd{\thmhead}{(#3)}{\curThmBraceOpen #3\curThmBraceClose }{}{}

\theoremstyle{theorem}
\newtheorem{lemma}[thm]{Lemma}
\newtheorem{proposition}[thm]{Proposition}
\newtheorem{theorem}[thm]{Theorem}
\newtheorem{corollary}[thm]{Corollary}
\theoremstyle{definition}
\newtheorem{notation}[thm]{Notation}
\newtheorem{example}[thm]{Example}
\newtheorem{construction}[thm]{Construction}
\newtheorem{assumption}[thm]{Assumption}

\AtBeginDocument{

}

\newcommand{\Perm}{\mathfrak{S}}

\newcommand{\At}{\mathbb{A}}

\title{Initial Algebras, Well-Founded Coalgebras, and Kőnig's Lemma} %

\author{Henning Urbat\lmcsorcid{0000-0002-3265-7168}}
\address{Friedrich-Alexander-Universität Erlangen-Nürnberg, Germany}
\email{henning.urbat@fau.de}
\urladdr{https://www8.cs.fau.de/people/henning-urbat/}

\author{Thorsten Wißmann\lmcsorcid{0000-0001-8993-6486}}
\email{thorsten.wissmann@fau.de}
\urladdr{https://thorsten-wissmann.de}

\keywords{\Konig's Lemma, Well-Founded Coalgebra, Recursive Coalgebra, Initial Algebra}

\FXRegisterAuthor{hu}{ahu}{HU} %
\FXRegisterAuthor{tw}{atw}{TW} %

\hypersetup{final}

\begin{document}

\maketitle

\begin{abstract}
\Konig's lemma is a fundamental result about trees with countless applications in mathematics and computer science. In contrapositive form, it states that if a tree is finitely branching and well-founded (i.e.\ has no infinite paths), then it is finite. We present a \emph{coalgebraic} version of \Konig's lemma featuring two dimensions of generalization: from finitely branching trees to coalgebras for a finitary endofunctor $H$, and from the base category of sets to a locally finitely presentable category $\C$, such as the category of posets, nominal sets, or convex sets. Our coalgebraic \Konig's lemma states that, under mild assumptions on $\C$ and $H$, every well-founded coalgebra for~$H$ is the directed join of its well-founded subcoalgebras with finitely generated state space -- in particular, the category of well-founded coalgebras is locally presentable. As applications, we derive versions of \Konig's lemma for graphs in a topos as well as for nominal and convex transition systems. Additionally, we show that the key construction underlying the proof gives rise to two simple constructions of the initial algebra (equivalently, the final recursive coalgebra) for the functor $H$: The initial algebra is both the colimit of all well-founded and of all recursive coalgebras with finitely presentable state space. Remarkably, this result holds even in settings where well-founded coalgebras form a proper subclass of recursive ones. The first construction of the initial algebra is entirely new, while for the second one our approach yields a short and transparent new correctness proof.
\end{abstract}

\section{Introduction}
In 1927, Dénes \Konig~\cite{konig1927} proved a result in graph theory, known as \emph{\Konig's lemma}, that turned out to be of fundamental importance. In its most popular formulation, the lemma asserts that every infinite, finitely branching tree contains an infinite path. While this statement looks almost obvious (with a classical, non-constructive mindset at least), it isolates an important selection principle that appears in different shapes and forms throughout mathematics and computer science, for instance, in proofs of compactness and completeness theorems in logic~\cite{teller89}, Ramsey theory~\cite{szekely23}, or the analysis of infinite runs and liveness properties of systems in verification~\cite{brv04} and automata theory~\cite{gmms14,isb16}. \Konig's lemma is moreover deeply tied with the foundations of set theory, specifically the axiom of countable choice~\cite{hr98,jech02}. Variants such as the \emph{weak \Konig's lemma} (the restriction to binary trees) and Brouwer's \emph{fan theorem}~\cite{brouwer27} have been extensively studied in (reverse) constructive mathematics and computability theory~\cite{di21,hirschfeldt14}.

\Konig's lemma applies to graphs, which are the simplest instance of \emph{state-based systems}, that is, systems with a notion of \emph{state space} and a notion of \emph{transition} between states. Such systems can be uniformly modelled at categorical generality using \emph{coalgebras}~\cite{jacobs16,rutten00}. A {coalgebra} is given by an endofunctor $H$ on a category $\C$ that specifies the {system type}, an object $C$ of $\C$ determining the state space, and a morphism $c\colon C\to HC$ determining the transition structure. By varying the base category $\C$ and the type functor $H$, countless different incarnations of state-based systems emerge in the coalgebraic framework. For example, graphs are precisely coalgebras for the power set functor $\Pow$ on the category of sets, non-deterministic register automata~\cite{KaminskiFrancez94} over an infinite alphabet $\mathds{A}$ of data values correspond to coalgebras for the functor $\{0,1\}\times \Pow^{\mathds{A}}$ on the topos of nominal sets~\cite{BojanczykEA14} (with power object functor $\Pow$), and (determinized) probabilistic automata are coalgebras for a functor that involves the convex power set functor on the category of convex sets~\cite{bss17}.

In the present work, we aim to address the following natural question:
\[ \emph{Which coalgebras (beyond graphs) feature a \Konig's lemma?} \]
To motivate our approach, let us revisit the case of graphs ($\Pow$-coalgebras). Note first that finitely branching graphs are precisely coalgebras for the \emph{finite} power set functor $\Powf$. We say that a graph is \emph{well-founded} if it contains no infinite paths. Then \Konig's lemma can be restated, in contrapositive form and in coalgebraic terminology, as follows:
\begin{equation}\label{eq:konig-pow} \emph{Every state of a well-founded $\Powf$-coalgebra lies in some finite subcoalgebra}. \end{equation}
The initial observation towards a generalization of \eqref{eq:konig-pow} is that the notion of well-foundedness is available for \emph{arbitrary} coalgebras~\cite{taylor99,amm25}: A coalgebra $c\colon C\to HC$ is \emph{well-founded} if it has no proper {cartesian} subcoalgebras (see \autoref{sec:prelim} for details). Moreover, $\Powf$ is the prototypical \emph{finitary} functor, that is, a functor that preserves filtered colimits. The first coalgebraic extension of \Konig's lemma, recently proved by Ad\'amek, Milius, and Moss~\cite[Prop.~9.2.19]{amm25}, replaces the finite power set functor $\Powf$ with an arbitrary finitary set functor~$H$:
\begin{equation}\label{eq:konig-fin-set} \emph{Every state of a well-founded $H$-coalgebra lies in some finite subcoalgebra}. \end{equation}
The proof from \emph{op.~cit.}\ exploits the fact that every coalgebra $c\colon C\to HC$ for a finitary set functor can be associated with a \emph{canonical graph} $\ol{c}\colon C\to \Powf C$ representing the transition structure of the coalgebra. In this way, \eqref{eq:konig-fin-set} reduces to the classical \Konig's lemma \eqref{eq:konig-pow}. 

The main contribution of our paper is a substantial generalization of \eqref{eq:konig-fin-set} to coalgebras over \emph{abstract} categories. This step involves two non-trivial challenges: First, canonical graphs are specific to coalgebras over the category of sets, so the reduction argument by Ad\'amek~et~al.~\cite{amm25} no longer works. Second, it is not clear what a `finite' coalgebra is in a general category. To address these challenges, we base our theory on \emph{\lfp} categories (which include, for instance, the categories of sets, nominal sets, and convex sets). The key feature of \lfplong categories is that they come with two well-behaved abstract notions of `finite' object, namely \emph{finitely presentable} and \emph{finitely generated} objects. Finite coalgebras thus naturally generalize to coalgebras with either finitely presentable or finitely generated state space, or \emph{fp/fg-carried coalgebras} for short. Our main result (\autoref{thm:koenig}) establishes \Konig's lemma at this level of generality:

\medskip\noindent \textsf{\bfseries Coalgebraic \Konig's Lemma.} \emph{Let $\C$ be a \lfplong category with monic coproduct injections, and let $H$ be a finitary functor on $\C$ preserving binary intersections. Then every well-founded $H$-coalgebra is the directed join of its fg-carried well-founded subcoalgebras.}

\medskip\noindent In \autoref{sec:applications}, we explore several instantiations of this result beyond the category of sets. Notably, we demonstrate that every \lfplong topos features a \Konig's lemma for graphs, and we derive versions of \Konig's lemma for nominal and convex transition systems. The latter systems are \emph{infinitely} branching, but their branching is finitely generated in an algebraic sense, which turns out to be enough to guarantee existence of infinite paths.

The core ingredient for our proof of the Coalgebraic \Konig's Lemma is a general construction on coalgebras, \emph{coproduct extension}~\cite{amv06,urbat17,lffiac,mpw16}, which allows to add new states to a coalgebra in a restricted way. The key observation is that coproduct extension preserves well-foundedness. This fact leads us to a further fundamental result on well-founded coalgebras, namely a new construction of the \emph{initial algebra} for the functor $H$ (\autoref{thm:wf-initial-algebra}):

\medskip\noindent \textsf{\bfseries Initial Algebra Theorem.} \emph{Under the assumptions of the Coalgebraic \Konig's Lemma, the initial algebra for the functor $H$ is the colimit of all fp-carried well-founded coalgebras.}

\medskip\noindent This theorem relates to a recent result by Wißmann and Milius~\cite{wm24} who expose the initial algebra as the colimit of all fp-carried \emph{recursive} coalgebras~\cite{taylor99}. Under the assumptions of our Initial Algebra Theorem, well-founded coalgebras form a (in some instances proper) subclass of recursive ones. Thus,
the Initial Algebra Theorem can be seen as an improvement of the corresponding result for recursive coalgebras because it operates with a smaller diagram, and more importantly, well-foundedness of a given coalgebra is typically much easier to prove than recursivity in settings where the two notions do not coincide. In addition, the idea underlying the proof of the Initial Algebra Theorem carries over to its recursive version and yields a transparent new proof of the latter that greatly simplifies the original argument~\cite{wm24}.

This paper is a revised and extended version of our contribution to CSL 2026~\cite{uw26}. In comparison to the conference paper, we have added the following new material:
\begin{enumerate}\item We provide full proofs of all our results. 
  \item We have re-organized and streamlined the material of \autoref{sec:initial-algebras} so that the two versions of the Initial Algebra Theorem for well-founded and recursive coalgebras now emerge uniformly from a general fixed point theorem (\autoref{thm:fixed-point}).
   \item We give an alternative short proof of the Initial Algebra Theorem for well-founded coalgebras (\autoref{thm:wf-initial-algebra}) that derives the theorem, under an additional assumption, as a direct corollary of the coalgebraic \Konig's lemma.
   \item We provide a formalization of the Initial Algebra Theorem in the proof assistant Agda, which emphasizes the constructive nature of this result. 
\end{enumerate}
\hunote{}

\section{Preliminaries}\label{sec:prelim}
We start with some background on \lfplong categories~\cite{adamek_rosicky_1994} and on well-founded and recursive coalgebras~\cite{amm25,taylor99}; see the cited textbooks for more details. Readers should be familiar with basic category theory~\cite{mac2013categories}, such as functors and (co)limits.

\subsection{Locally Finitely Presentable Categories}\label{sec:lfp}
The results of our paper apply to coalgebras over \lfplong categories. Informally, these are categories where every object can be approximated from below by suitably `finite' objects. For example, the category of sets is \lfplong because every set is a union of finite sets, and the category of vector spaces is \lfplong because every vector space is a  union of finite-dimensional spaces. The categorical notion of finiteness is expressed abstractly using filtered colimits.

A category $I$ is \emph{filtered} if (i) $I$ is non-empty, (ii) for any two objects $i,j\in I$ there exists a cospan $i\to k \leftarrow j$ in $I$, and (iii) for every parallel pair $f,g\colon i\to j$ there exists a morphism $h\colon j\to k$ such that $h\circ f = h\circ g$. For example, finitely cocomplete categories are filtered. A poset (viewed as a category) is filtered iff it is directed, i.e.\ every finite subset has an upper bound. A \emph{filtered diagram} $D\colon I\to \C$ is a diagram whose scheme $I$ is filtered, and a \emph{filtered colimit} is a colimit of a filtered diagram. A functor is \emph{finitary} if it preserves filtered colimits.

\begin{example}[Filtered Colimits in $\Set$]\label{ex:filtered-set}
In the category $\Set$ of sets and functions, the colimit of a filtered diagram $D\colon I\to \Set$ is given by $(\coprod_i D_i)/{\approx}$ where $\approx$ is the equivalence relation given by
\[  (x,i)\approx (y,j)\;\;\;\text{iff}\;\;\; \text{there exist morphisms $f\colon i\to k$ and $g\colon j\to k$ in $I$ with $Df(x)=Dg(y)$}.\] 
The colimit injection $c_i\colon D_i\to (\coprod_i D_i)/{\approx}$ sends $x\in D_i$ to the equivalence class of $(x,i)$. An endofunctor $H$ on $\Set$ is finitary iff for every set $X$ and every $t\in HX$, there exists a finite subset $m\colon M\subto X$ such that $t\in Hm[HM]$~\cite[Rem.\ 3.14]{amsw19}.
\end{example}

An object $X$ of a category $\C$ is \emph{finitely presentable} if its hom-functor $\C(X,-)\colon \C\to \Set$ is finitary. More explicitly, this means that given a filtered diagram $D\colon I\to \C$ with colimit cocone $(c_i\colon D_i\to C)_{i\in I}$, every morphism from $X$ to $C$ \emph{factorizes essentially uniquely through the cocone $(c_i)$}, that is, the following two statements hold:
\begin{enumerate}
\item every morphism $f\colon X\to C$ factorizes as $f=c_i\circ g$ for some $i\in I$ and $g\colon X\to D_i$;
\item for any two morphisms $g,g'\colon X\to D_i$ with $i\in I$ and $c_i\circ g = c_i\circ g'$, there exists $k\colon i\to j$ in $I$ such that $Dk\circ g = Dk\circ g'$.
\end{enumerate}
Besides finitely presentable objects, there is also the weaker notion of \emph{finitely generated} object. An object $X$ of $\C$ is \emph{finitely generated} if for every filtered diagram $D\colon I\to \C$ whose colimit cocone $(c_i\colon D_i\monoto C)_{i\in I}$ consists of monomorphisms, every morphism from $X$ to $C$ factorizes through some $c_i$. (Here the essential uniqueness of the factorization comes for free because $c_i$ is monic.) We let \[\C_\fp,\C_\fg\hookrightarrow \C\] denote the full subcategories given by finitely presentable and finitely generated objects, respectively. Note that $\C_\fp\seq \C_\fg$, but generally finitely presentable and finitely generated objects do not coincide. 

\begin{notheorembrackets}
\begin{lemma}\label{lem:fp-fp-fin-colimits} The subcategories $\C_\fp$, $\C_\fg \subto \C$ have the following closure properties:
\begin{enumerate}
\item Both $\C_\fp$ and $\C_\fg$ are closed under finite colimits~\cite[Prop.~1.3]{adamek_rosicky_1994}. 
\item $\C_\fg$ is closed under strong quotients (represented by strong epimorphisms)~\cite[Prop.~1.69]{adamek_rosicky_1994}.
\end{enumerate}
\end{lemma}
\end{notheorembrackets}

A category $\C$ is \emph{\lfplong} if it is cocomplete, and there exists a (small) set~$\A$ of finitely presentable objects such that every object of $\C$ is a filtered colimit of objects in $\A$, that is, a colimit of some filtered diagram $D\colon I\to \C$ such that $D_i\in \A$ for all $i\in I$. Informally, this says that every object can be constructed from `finite' objects.

\begin{example}[Locally Finitely Presentable Categories]\label{ex:categories} 
\begin{figure*}[t]
\centering
\def\arraystretch{1.0}
\begin{tabular}{@{}p{1.5cm}p{2.5cm}p{3.2cm}p{5.5cm}@{}}
\toprule
$\C$ & Objects & Morphisms & $\C_\fp$ ($=\C_\fg$) \\
\midrule
$\Set$ & sets & functions & finite sets~ \\
$\Nom$ & nominal sets &  equivariant maps & orbit-finite sets \\
$\Conv$ & convex sets & affine maps & finitely generated convex sets \\
\bottomrule
\end{tabular}
\caption{\Lfplong categories.}\label{fig:categories}
\end{figure*}
In our applications (\autoref{sec:applications}) we consider the \lfp categories of \autoref{fig:categories}. Their finitely presentable and finitely generated objects coincide and are described in the last column. The categories $\Nom$ of nominal sets and $\Conv$ of convex sets are discussed in more detail in \autoref{sec:nom} and \autoref{sec:convex}. Further examples of \lfp categories (see e.g.~\cite{adamek_rosicky_1994}) include:
\begin{enumerate}\item all categories of algebras for an equational theory (such as monoids, rings, vector spaces); \item all categories of relational structures axiomatized by a finitary relational Horn theory (such as partially ordered sets); \item  all presheaf categories $\Set^\D$ over a small category $\D$.
\end{enumerate} 
 In all these categories, filtered colimits are formed like in $\Set$ (\autoref{ex:filtered-set}).
\end{example}

\begin{rem}\label{rem:loc-pres}
There is a more general notion of \emph{locally $\lambda$-presentable category}, for a regular cardinal $\lambda$, where the role of filtered colimits is taken over by \emph{$\lambda$-filtered colimits}. Here a diagram scheme is \emph{$\lambda$-filtered} if condition (ii) in the above definition of a filtered category is strengthened to ($\text{ii}_\lambda$): Every non-empty set of objects of cardinality less than $\lambda$ has a cocone over it, and likewise, condition (iii) is strengthened to to families of parallel arrows of cardinality less than $\lambda$. \Lfp categories correspond to the case $\lambda=\omega$. A category is \emph{locally presentable} if it is locally $\lambda$-presentable for some $\lambda$. An example of a locally ($\omega_1$-)presentable category that is not locally \emph{finitely} presentable is the category of metric spaces and non-expansive maps. The results of our paper currently do not extend to arbitrary locally presentable categories: we use several results on well-founded coalgebras (see below) which rest on properties specific to ($\omega$-)filtered colimits.
\end{rem}

\Lfp categories are, in many respects, convenient and well-behaved structures. The lemmas below collect a number of core properties, whose proofs can be found in Ad\'amek and Rosick\'y~\cite{adamek_rosicky_1994} (Remark 1.56, Proposition 1.61, Proposition 1.62, Exercise~1.o).

\begin{lemma}\label{lem:lfp-complete-wp}
Every \lfp category $\C$ is complete and well-powered, and every morphism in $\C$ has a (strong epi, mono)-factorization.
\end{lemma}
Since a \lfp category $\C$ is cocomplete and has (strong epi, mono)-factorizations, the poset $\Sub(X)$ of subobjects of an object $X$ forms a complete lattice: The join of a family of subobjects $m_i\colon X_i\monoto X$ ($i\in I$) is the subobject $m\colon \bigvee_{i\in I} X_i \monoto X$ given by the (strong epi, mono)-factorization of the morphism $[m_i]_{i\in I}$:
\[ 
\begin{tikzcd}[outer sep=0pt]
\coprod_{i\in I} X_i
  \ar[shiftarr={yshift=15}]{rr}{[m_i]_{i\in I}}
  \ar[two heads]{r}
  & \bigvee_{i\in I} X_i \ar[tail]{r}{m}
  & X
   \end{tikzcd}
\]
\emph{Directed} joins of subobjects can be alternatively computed by just forming their colimit. In more detail, a directed family $m_i\colon X_i\monoto X$ ($i\in I$) of subobjects can be regarded as a cocone over the directed diagram $D\colon I\to \C$, $i\mapsto X_i$, where $I$ is ordered by $i\leq j$ iff $m_i\leq m_j$ in the usual partial order of subobjects. Then the following statement holds:

\begin{lemma}\label{lem:lfp-directed-unions}
Let $\C$ be \lfp. For any directed family $m_i\colon X_i\monoto X$ ($i\in I$) of subobjects of an object $X$, its join in $\Sub(X)$ is the colimit of the induced diagram.
\end{lemma}

We also have a useful characterization of filtered colimits in general:

\begin{lemma}\label{lem:filtered-colimits}
Let $D\colon I\to \C$ be a filtered diagram in a \lfp category~$\C$. A cocone $(c_i\colon D_i\to C)_{i\in I}$ over $D$ forms a colimit cocone iff for every $X\in \C_\fp$, every morphism from $X$ to $C$ factorizes essentially uniquely through $(c_i)$.
\end{lemma}
Note that the left-to-right implication holds by the definition of a finitely presentable object. The reverse implication is the interesting part of the lemma. 

Regarding finitary functors, the following result generalizing \autoref{ex:filtered-set} holds:

\begin{notheorembrackets}
\begin{lemma}[{\cite[Thm.~3.4]{amsw19}}]\label{lem:finitary-criterion}
Let $\C$ be a \lfp category with $\C_\fp = \C_\fg$, and let $H$ be a functor that preserves monomorphisms. Then $H$ is finitary iff for every object~$X$ of $\C$, every finitely generated subobject $s\colon S\monoto HX$ factorizes through $Hm\colon HM\monoto HX$ for some finitely generated subobject $m\colon M\monoto X$.  
\end{lemma}
\end{notheorembrackets}
This  applies, in particular, to mono-preserving functors on the categories of \autoref{fig:categories}.

Lastly, if a functor is not finitary, it can be canonically restricted to a finitary functor. In the following, let $[\C,\C]$ denote the category of endofunctors on $\C$ and natural transformations, $[\C,\C]_{\mathsf{fin}}$ its full subcategory given by finitary endofunctors, and $\C_\fp/X$ the full subcategory of the slice category $\C/X$ whose objects are morphisms $(f\colon P\to X)$ of $\C$ with $P\in \C_\fp$.

\begin{notheorembrackets}
\begin{lemma}[{\cite[Cor.~2.8]{amv03}}]\label{lem:finitary-coreflection}
If $\C$ is \lfp, then $[\C,\C]_{\mathsf{fin}}$ is a coreflective subcategory of $[\C,\C]$. The coreflection of $H\colon \C\to \C$ is the left Kan extension $H_\omega\colon \C\to\C$ of the functor $HI$ along the inclusion $I\colon \C_\fp\subto \C$. More explicitly,
\[ H_\omega X=\colim D_X \] for the diagram $D_X\colon \C_\fp/X\to \C$ mapping $(f\colon P\to X)\in \C_\fp/X$ to $HP$.   
\end{lemma}
\end{notheorembrackets}
The functor $H_\omega$ is called the \emph{finitary coreflection} of $H$. For example, the finitary coreflection of the power set functor $\Pow\colon \Set\to \Set$ is the finite power set functor $\Powf\colon \Set\to \Set$.

\subsection{Coalgebras and Algebras} 

Coalgebras~\cite{jacobs16,rutten00} form a uniform categorical abstraction of state-based transition systems, such as graphs, labelled transition systems, Markov chains, and various kinds of automata. The idea is to model the state space of a system by an object of a suitable category $\C$, and its transition type by an endofunctor on that category. Formally, a \emph{coalgebra} for an endofunctor $H\colon \C\to \C$ is a pair $(C,c)$ consisting of an object $C$ (the
\emph{state space}) and a morphism $c\colon C\to HC$ (its
\emph{structure}). A \emph{morphism} $h\colon (C,c)\to (D,d)$ of coalgebras is a morphism $h\colon C\to D$ of $\C$ such that the square below commutes:
\[ 
\begin{tikzcd}
C \ar{r}{c} \ar{d}[swap]{h} & HC \ar{d}{Hh} \\
D \ar{r}{d}& HD 
\end{tikzcd}
\]
We let $\Coalg(H)$ denote the category of coalgebras for $H$ and their morphisms. 

\emph{Subcoalgebras} of a coalgebra $(C,c)$ are represented by coalgebra morphisms  $m\colon (S,s)\monoto (C,c)$ with $m$ monic in $\C$, and \emph{strong quotients} of $(C,c)$ are represented by coalgebra morphisms $e\colon (C,c)\epito (Q,q)$ with $e$ strongly epic in $\C$. If $H$ preserves monomorphisms, a subcoalgebra is uniquely determined by the subobject $S$. If moreover $\C$ has (strong epi, mono)-factorizations, then every coalgebra morphism $h\colon (C,c)\to (D,d)$ factorizes uniquely as a strong quotient followed by a subcoalgebra: \[\begin{tikzcd} (C,c) \ar[shiftarr={yshift=20}]{rr}{h} \ar[two heads]{r}{e} & (S,s) \ar[tail]{r}{m} & (D,d) \end{tikzcd}\] This factorization is called the \emph{image factorization} of $h$, and $m$ is the \emph{image} of $h$.

\begin{example}[Graphs]\label{ex:coalgebras}
 A coalgebra for the power set functor $\Pow$ on $\Set$ is precisely a (directed) graph. A coalgebra for the finite power set functor $\Powf$ (the finitary coreflection of $\Pow$) corresponds to a \emph{finitely branching} graph, i.e.\ one where every node has only finitely many outgoing edges. Given a graph $(C,c)$, we write $x\to y$ iff $y\in c(x)$. 
A subcoalgebra of $(C,c)$ corresponds to a subset $S\seq C$ closed under successors: \[\text{$x\in S$}\quad\text{and}\quad \text{$x\to y$}\qquad \text{implies}\qquad \text{$y\in S$}.\]
\end{example}

\begin{example}[Automata and Transition Systems]
Various types of automata and transition systems can be modelled as coalgebras for a set functor $H$, including \begin{enumerate}\item deterministic automata ($HX=\{0,1\}\times X^A$ for a fixed finite input alphabet $A$); \item  non-deterministic automata ($HX=\{0,1\}\times (\Pow X)^A$); \item labelled transition systems ($HX = (\Pow X)^A$);\item  labelled Markov chains ($HX=(\mathcal{D}X)^A$ where $\mathcal{D}$ is the finite distribution functor). \end{enumerate} 
\end{example}

Examples of coalgebras beyond the category of sets are discussed later, e.g.\ graphs in a topos (\autoref{sec:topos}), coalgebras in nominal sets which model automata and transition systems over
  infinite alphabets (\autoref{sec:nom}), and 
  coalgebras in convex sets, modelling systems that combine probabilistic and non-deterministic branching
  (\autoref{sec:convex}).

We refer the reader to Jacobs~\cite{jacobs16} or Rutten~\cite{rutten00} for general introductions to the rich theory of coalgebras. For our purposes, we recall two basic results. The first one asserts that colimits of coalgebras are formed in the underlying category:

\begin{notheorembrackets}
\begin{lemma}[{\cite[Prop.~4.7]{rutten00}}]\label{lem:colimits-coalg}
For every endofunctor $H\colon \C\to \C$, the forgetful functor from $\Coalg(H)$ to $\C$ creates colimits.
\end{lemma}
\end{notheorembrackets}

The second result concerns coalgebras with `finite' state space. A coalgebra is \emph{fp-carried} if its state space is finitely presentable, and \emph{fg-carried} if its state space is finitely generated. We denote the corresponding full subcategories by \[\Coalg_\fp(H)\subto \Coalg_\fg(H)\subto \Coalg(H).\]

\begin{notheorembrackets}
\begin{lemma}[{\cite[Lem.~3.2]{ap04}}]\label{lem:fp-fg-coalg} Let $H\colon \C\to \C$ be a finitary endofunctor.
\begin{enumerate}
\item Every fp-carried coalgebra is a finitely presentable object of $\Coalg(H)$.
\item If $H$ preserves monos, every fg-carried coalgebra is a finitely generated object of $\Coalg(H)$.
\end{enumerate}
\end{lemma}
\end{notheorembrackets}

Dually to the notion of coalgebra, \emph{algebras} for an endofunctor yield a categorical abstraction of algebraic structures. An \emph{algebra} $(A,a)$ for $H\colon \C\to \C$ is given by an object~$A$ (its \emph{carrier}) and a morphism $a\colon H A\to A$ (its \emph{structure}). A
\emph{morphism} $h\colon (A,a) \to (B,b)$ of algebras is a morphism
$h\colon A\to B$ of~$\gcat$ such that the square 
\[ 
\begin{tikzcd}
HA \ar{r}{a} \ar{d}[swap]{Hh} & A \ar{d}{h} \\
HB \ar{r}{b}& B
\end{tikzcd}
\]  
commutes.  Algebras
for $H$ and their morphisms form a category $\Alg(H)$. An \emph{initial algebra} $(I,i)$ is an initial object of $\Alg(H)$. By Lambek's lemma~\cite{lambek68}, its structure $i\colon HI\to I$ is an isomorphism in $\C$.

\begin{example}[$\Sigma$-Algebras]
The prime example of functor algebras are algebras for a
signature. An \emph{(algebraic) signature} is given by a set~$\Sigma$
of \emph{operation symbols} and a map $\ar\colon \Sigma\to \Nat$
associating to every $\f\in \Sigma$ its \emph{arity}. Every
signature~$\Sigma$ induces the polynomial endofunctor $H_\Sigma X = \coprod_{\f\in \Sigma} X^{\ar(\f)}$ on $\Set$. An algebra for $H_\Sigma$ is
precisely an algebra for the signature~$\Sigma$: a set $A$ with an operation $\f^A\colon A^n\to A$ for every $n$-ary operation symbol $\f\in \Sigma$. The initial algebra for $H_\Sigma$ is the algebra of closed terms formed over the signature $\Sigma$.
\end{example}

\subsection{Well-Founded and Recursive Coalgebras}\label{sec:prelim:wf-rec}
The concept of `having no infinite paths' can be extended from graphs to arbitrary coalgebras~\cite{taylor99}. A subcoalgebra $m\colon (S,s)\monoto (C,c)$ of a coalgebra $(C,c)$ is \emph{cartesian} if the commutative diagram below is a pullback:
\begin{equation}\label{eq:cartesian}  \begin{tikzcd}
  S \pullbackangle{-45} \ar{r}{s} \ar[tail]{d}[swap]{m}  & HS \ar{d}{Hm} \\
  C \ar{r}{c} & HC 
  \end{tikzcd}
\end{equation}
  A coalgebra $(C,c)$ is \emph{well-founded} if it has no proper cartesian subcoalgebras: for every cartesian subcoalgebra $m\colon (S,s)\monoto (C,c)$, the monomorphism $m$ is an isomorphism.

\begin{example}[Graphs]\label{ex:wfCoalg}
Let $(C,c)$ be a $\Pow$-coalgebra, viewed as a graph. We note first that a subset $S\seq C$ carries a cartesian subcoalgebra iff, for all $x\in C$,
\begin{equation}\label{eq:cartesian-pow}
x\in S \quad \iff \quad \text{all successors of $x$ lie in $S$.}
\end{equation}
Indeed, the left-to-right implication says that the square \eqref{eq:cartesian} commutes ($S$ is a subcoalgebra), and the right-to-left implication says that it is a pullback. From this, it follows that
\[ \text{$(C,c)$ is well-founded} \quad\iff\quad \text{$(C,c)$ has no infinite paths}. \]
 For the left-to-right implication, suppose that $(C,c)$ is well-founded. Then the set $S\seq C$ of all states that lie on no infinite path is a cartesian subcoalgebra by \eqref{eq:cartesian-pow}, and so $S=C$. For the right-to-left implication, we argue contrapositively. Suppose that $(C,c)$ is not well-founded, and let $S\subsetneq C$ be a proper cartesian subcoalgebra. Pick $x_0\in C\setminus S$ arbitrarily. By \eqref{eq:cartesian-pow}, some successor $x_1$ of $x_0$ lies in $C\setminus S$. By \eqref{eq:cartesian-pow} again, some successor $x_2$ of $x_1$ lies in $C\setminus S$. Repeating this argument yields an infinite path $x_0\to x_1\to x_2\to \cdots$.
\end{example}

\begin{example}[Coalgebras for Set Functors]
For every functor $H\colon \Set\to \Set$ that preserves wide intersections (i.e.\ wide pullbacks of monomorphisms) and every set $X$, there is map \[\supp_X\colon HX\to \Pow X\] that sends an element $t\in HX$ to its \emph{support}, the least subset $m\colon M\subto X$ such that $t\in Hm[HM]$. Every coalgebra $(C,c)$ for $H$ thus induces a coalgebra \[ C \xto{c} HC \xto{\supp_C} \Pow C\] for~$\Pow$, the \emph{canonical graph} of $(C,c)$. A coalgebra for $H$ is well-founded iff its canonical graph is well-founded~\cite[Rem.~6.3.4]{taylor99}.
\end{example}
One important application of well-founded graphs is the principle of \emph{well-founded induction} and \emph{well-founded recursion}. This principle can be understood at the level of coalgebras~\cite{taylor99}. Given a coalgebra $(C,c)$ and an algebra $(A,a)$ for the same functor $H$, a \emph{coalgebra-to-algebra morphism} from $(C,c)$ to $(A,a)$ is a morphism $h\colon C\to A$ making 
\[
\begin{tikzcd}
C \ar{r}{h} \ar{d}[swap]{c} & A \\
HC \ar{r}{Hh} & HA \ar{u}[swap]{a}
\end{tikzcd}
\]   
commute. A coalgebra $(C,c)$ is \emph{recursive} if for every algebra $(A,a)$ there exists a unique coalgebra-to-algebra morphism from $(C,c)$ to $(A,a)$.

\begin{example}[Graphs]
Every well-founded $\Pow$-coalgebra $(C,c)$ is recursive. Indeed, this amounts to saying that for every $a\colon \Pow A \to A$, there is unique map $h\colon C\to A$ such that \[h(x)=a(h[c(x)]) \qquad \text{for all $x\in C$}.\] This is precisely the principle of \emph{well-founded recursion}. The special case where $a\colon \Pow\{0,1\}\to \{0,1\}$ sends $\emptyset$, $\{1\}$ to $1$ and $\{0\}$, $\{0,1\}$ to $0$ corresponds to \emph{well-founded induction}: given a predicate $P\colon C\to \{0,1\}$, one has
\[ \big(\,\forall x\in C.\, (\forall y\in c(x).\, P(y)) \implies P(x)\,\big) \quad\implies\quad \big(\,\forall x\in C.\, P(x)\,\big). \] 
\end{example}
\noindent

\begin{notation}
In the following, we denote by
\[\Coalg_{\fp,\wf}(H),\, \Coalg_{\fg,\wf}(H),\, \Coalg_\wf(H),\, \Coalg_{\fp,\rec}(H),\, \Coalg_{\fg,\rec}(H),\, \Coalg_\rec(H) \] 
the full subcategories of $\Coalg(H)$ given by all (fp-carried, fg-carried)  well-founded coalgebras, and analogously for recursive coalgebras. 
\end{notation}

We conclude with some general properties of well-founded and recursive coalgebras. First, colimits of recursive coalgebras are computed in the underlying category~\cite[Prop.~7.1.9]{amm25}:
\begin{lemma}\label{lem:rec-coalg-colimits}
The subcategory $\Coalg_\rec(H)\subto \Coalg(H)$ is closed under colimits.
\end{lemma}

Regarding well-founded coalgebras, we need the results of the proposition below. We state them for a \lfp base category $\C$, our setting of interest; however, they hold under slightly weaker assumptions~\cite[Cor.~8.4.3, Thm.~8.5.3]{amm25}. In the following, an initial object $0$ is called \emph{simple} if for every object $X$ the unique morphism $0\to X$ is monic.

\begin{proposition}\label{prop:wf} 
Let $\C$ be a \lfp category with a simple initial object, and let $H\colon \C\to \C$ be an endofunctor that preserves monomorphisms.
\begin{enumerate}
\item\label{prop:wf:closed-colim-quotients} The subcategory $\Coalg_\wf(H)\subto \Coalg(H)$ is closed under colimits and strong quotients.
\item\label{prop:wf:wf-to-rec} Every well-founded coalgebra is recursive, that is, $\Coalg_\wf(H)\seq \Coalg_\rec(H)$.
\end{enumerate}
\end{proposition}

\section{\Konig's Lemma for Coalgebras}\label{sec:konig}
We are prepared to work towards the core result of our paper, \Konig's lemma for coalgebras over a \lfp category (\autoref{thm:koenig}). Let us first fix the global setup:

\begin{assumption}\label{asm}
For the remainder of this paper, we fix a finitary endofunctor $H$ on a \lfplong category $\C$. 
\end{assumption}
Most of our results require the following two additional conditions on $\C$ and $H$:

\medskip\noindent (\injmono) The injections $\inl\colon A\to A+B$ and $\inr\colon B\to A+B$ of binary coproducts in $\C$ are monic.

\medskip\noindent (\hpresint) The functor $H$ preserves binary intersections (i.e.\ pullbacks of monomorphisms).

\medskip\noindent We will tag all results below with the conditions they require (on top of the global conditions given by \autoref{asm}). For example, \autoref{prop:wf-coalgebra-coproduct} requires both (\injmono) and (\hpresint).

\begin{rem}\label{rem:conditions}
The two extra conditions (\injmono) and (\hpresint) are fairly mild:
\begin{enumerate}
\item Condition (\injmono) holds in every extensive category with finite limits~\cite[Prop.~3.3]{cw93}. Recall that a category $\C$ is \emph{extensive}~\cite{cw93} if it has finite coproducts and for any pair of objects~$A$,~$B$ the functor on slice categories \[\C/A\times \C/B \xto{+} \C/(A+B), \quad (f\colon X\to A, g\colon Y\to B)\mapsto (f+g\colon X+Y\to A+B),\] is an equivalence of categories. Informally, this means that coproducts behave like disjoint unions of sets. The categories of sets, posets, and nominal sets are extensive, as is every elementary topos~\cite{Johnstone}. There are also many non-extensive categories satisfying (\injmono), including the categories of vector spaces or convex sets (\autoref{sec:convex}).
\item\label{rem:conditions:set-func-pres-int} Condition (\hpresint) almost comes for free for sets functors. By a classical result due to Trnkov\'a~\cite{trnkova71} (see also~\cite[Prop.~8.1.13]{amm25}), every functor~$H\colon \Set\to\Set$ can be modified, by changing its action on the empty set and functions with empty domain, to a functor~$\widetilde{H}$ that preserves binary intersections. Since $\Coalg(H)\cong \Coalg(\widetilde{H})$, one can thus assume that $H$ preserves binary intersections without loss of generality for coalgebraic results.
\end{enumerate} 
\end{rem}

\begin{rem}\label{rem:int-implies-mono}
\begin{enumerate}
\item Condition ($\injmono$) implies that the initial object $0$ is simple: given $X\in \C$, the unique morphism $0\to X$ is the coproduct injection $\inl\colon 0\to 0+X\cong X$, hence monic. 
\item Condition (\hpresint) implies that $H$ preserves monomorphisms. This follows from the fact that a morphism $m\colon A\to B$ is monic iff the pair $\id_A,\id_A$ is a pullback of $m,m$.
\end{enumerate}

\end{rem}

We start with a simple construction, the \emph{coproduct extension} of a coalgebra. It appeared earlier in the theory of iterative algebras~\cite{amv06,urbat17} and underlies both our proof of the coalgebraic \Konig's lemma and the initial algebra constructions presented in \autoref{sec:initial-algebras}. 

\begin{construction}[Coproduct Extension]\label{cons:cp}
The \emph{coproduct extension} of an $H$-coalgebra $(C,c)$ by a morphism $p\colon X\to HC$ is the coalgebra $(C+X,c_p)$ whose structure is given by
\[ c_p \, \equiv \, (\,C+X \xto{[c,p]} HC \xto{H\inl} H(C+X)\,).   \]
\end{construction}
Let us immediately note an expected property of coproduct extensions:

\begin{lemma}\label{lem:coproduct-inl-coalg-mor} The coproduct injection $\inl\colon (C,c)\to (C+X,c_p)$ is a coalgebra morphism.
\end{lemma}

\begin{proof}
This is shown by the commutative diagram below:
\[
\begin{tikzcd}[baseline=(HC.base)]
C \ar{dr}{c} \ar{r}{\inl} \ar{dd}[swap]{c} & C+X \ar[shiftarr={xshift=15mm}]{dd}{c_p} \ar{d}{[c,p]} \\
& HC \ar{d}{H\inl} \\
|[alias=HC]|
HC \ar{r}{H\inl} & H(C+X)
\end{tikzcd}
\tag*{\qedhere}
\]
\end{proof}

Informally, a coproduct extension adds to a given coalgebra $(C,c)$ a set $X$ of new states and lets every new state transition arbitrarily into states from $C$, while leaving transitions of states in $C$ unchanged. This construction should not introduce any infinite paths to the given coalgebra, which is confirmed by the following result: 

\begin{proposition}[\injmono, \hpresint]\label{prop:wf-coalgebra-coproduct}
For every well-founded coalgebra $(C,c)$ and every morphism $p\colon X\to HC$, the coproduct extension $(C+X,c_p)$  is a well-founded coalgebra.
\end{proposition}

\twnote{}

\begin{proof}
Let $m\colon (S,s)\monoto (C+X,c_p)$ be a cartesian subcoalgebra. We need to prove that $m$ is an isomorphism. Form the two intersections in $\C$ shown below; note that the coproduct injections $\inl$ and $\inr$ are monic by assumption (\injmono).\\

\begin{minipage}{.45\textwidth}
\begin{equation}\label{eq:sc-int}
\begin{tikzcd}
S\cap C \pullbackangle{-45} \ar[tail]{d}[swap]{l_{S\cap C}} \ar[tail]{r}{r_{S\cap C}} & C \ar[tail]{d}{\inl} \\
S \ar[tail]{r}{m} & C+X
\end{tikzcd}
\end{equation}
\end{minipage}
\begin{minipage}{.45\textwidth}
\begin{equation}\label{eq:sx-int}
\begin{tikzcd}
S\cap X \pullbackangle{-45} \ar[tail]{d}[swap]{l_{S\cap X}} \ar[tail]{r}{r_{S\cap X}} & X \ar[tail]{d}{\inr} \\
S \ar[tail]{r}{m} & C+X
\end{tikzcd}
\end{equation}
\end{minipage}\\

\noindent Below we shall prove the following:
\begin{enumerate}
\item\label{s1} $r_{S\cap C}$ is an isomorphism.
\item\label{s2} $r_{S\cap X}$ is an isomorphism.
\end{enumerate}
These statements imply that $m$ is a split epimorphism:
\[ m\circ [l_{S\cap C} \circ r_{S\cap C}^{-1}, l_{S\cap X}\circ r_{S\cap X}^{-1}] = [\inl,\inr]=\id_{C+X}, \]
and since $m$ is also monic, it is an isomorphism as required. 

It thus only remains to verify the above two claims.

\medskip\noindent \emph{Proof of (\ref{s1})}: The following diagram commutes:
\begin{equation}\label{eq:comm-diag}
\begin{tikzcd}
S\cap C \ar[tail]{dd}[swap]{r_{S\cap C}} \ar[tail]{rr}{l_{S\cap C}}  && S \ar{d}{s} \ar[tail]{dl}[swap]{m} \\
& C+X \ar{d}{[c,p]} \ar{dr}{c_p} & HS \ar[tail]{d}{Hm} \\
C \ar{ur}{\inl} \ar{r}{c} & HC \ar{r}{H\inl} & H(C+X)
\end{tikzcd}
\end{equation}
Therefore, by \eqref{eq:sc-int} and since $H$ preserves intersections, there exists a unique morphism $s_C$ making the left-hand and upper part of the diagram below commute:
\begin{equation}\label{eq:def-ol-sc}
\begin{tikzcd}
S\cap C \ar[tail]{dd}[swap]{r_{S\cap C}} \ar[tail]{rr}{l_{S\cap C}}  \ar[dashed]{dr}{s_C} && S \ar{d}{s} \\
& H(S\cap C)  \pullbackangle{-45}\ar[tail]{r}{Hl_{S\cap C}} \ar[tail]{d}{Hr_{S\cap C}} & HS \ar[tail]{d}{Hm} \\
C \ar{r}{c} & HC \ar{r}{H\inl} & H(C+X)
\end{tikzcd}
\end{equation}
We claim that $r_{S\cap C}\colon (S\cap C, s_C) \monoto (C,c)$ is a cartesian subcoalgebra, which means that the following diagram is a pullback:
\begin{equation}\label{eq:sc-pullback}
\begin{tikzcd}
S \cap C \ar[tail]{r}{s_C} \ar[tail]{d}[swap]{r_{S\cap C}} \pullbackangle{-45} & H(S\cap C) \ar[tail]{d}{Hr_{S\cap C}}\\
C \ar{r}{c} & HC
\end{tikzcd}
\end{equation}
To see this, let $f\colon Z\to H(S\cap C)$ and $g\colon Z\to C$ be morphisms with
\begin{equation}\label{eq:def-fg}
Hr_{S\cap C}\circ f = c\circ g.
\end{equation}
Then the diagram below commutes:
\begin{equation}\label{eq:comm-diag2}
\begin{tikzcd}
Z \ar{rrr}{f} \ar{dd}[swap]{g} & & & H(S\cap C) \ar[tail]{d}{Hl_{S\cap C}} \ar[tail]{ddl}[swap]{Hr_{S\cap C}}  \\
& & & HS \ar[tail]{d}{Hm} \\
  C \ar[rounded corners,to path={
    ([xshift=2mm]\tikztostart.north)
    -- ++(0,4mm)
    -- ([yshift=4mm,xshift=-1mm]\tikztotarget.north) \tikztonodes
    -- ([xshift=-1mm]\tikztotarget.north)
  }]{rr}{c} \ar{r}{\inl} & C+X \ar{r}{[c,p]} & HC \ar{r}{H\inl} & H(C+X)
\end{tikzcd}
\end{equation}
Therefore, there exists a unique morphism $j$ making the left-hand and upper part of the following diagram commute:
\begin{equation}\label{eq:def-j}
\begin{tikzcd}
Z \ar[dashed]{dr}{j} \ar{rrr}{f} \ar{dd}[swap]{g} & & & H(S\cap C) \ar[tail]{d}{Hl_{S\cap C}}  \\
& S \pullbackangle{-45} \ar[tail]{d}[swap]{m} \ar{rr}{s} & & HS \ar[tail]{d}{Hm} \\
C \ar{r}{\inl} & C+X \ar{r}{[c,p]} & HC \ar{r}{H\inl} & H(C+X)
\end{tikzcd}
\end{equation}
This in turn yields a unique morphism $k$ making the left-hand and upper part of the diagram
\begin{equation}\label{eq:def-k}
\begin{tikzcd}
Z \ar[dashed]{dr}{k} \ar[bend left=20]{drr}{g} \ar[bend right=20]{ddr}[swap]{j} & &  \\
& S\cap C \pullbackangle{-45} \ar[tail]{r}{r_{S\cap C}} \ar[tail]{d}[swap]{l_{S\cap C}} & C \ar[tail]{d}{\inl} \\
& S \ar[tail]{r}{m} & C+X
\end{tikzcd}
\end{equation}
commute. Then $k$ also makes the left-hand and upper part of the diagram below commute:
\begin{equation}\label{eq:k-pullback}
\begin{tikzcd}
Z \ar[bend left=20]{drr}{f} \ar[bend right=20]{ddr}[swap]{g} \ar{dr}{k} && \\
& S \cap C \ar[tail]{r}{s_C} \ar[tail]{d}[swap]{r_{S\cap C}} & H(S\cap C) \ar[tail]{d}{Hr_{S\cap C}}\\
& C \ar{r}{c} & HC
\end{tikzcd}
\end{equation}
Indeed, the left-hand part commutes by definition of $k$, and the upper part commutes because it does so when post-composed with the monomorphism $Hr_{S\cap C}$. Moreover, since $r_{S\cap C}$ is monic, $k$ is the unique morphism making both parts commute. This proves that \eqref{eq:sc-pullback} is a pullback, as claimed. Since $(C,c)$ is well-founded, we conclude that $r_{S\cap C}$ is an isomorphism.

\medskip\noindent \emph{Proof of (\ref{s2})}:
The following diagram commutes:
\begin{equation}\label{eq:comm-diag3}
\begin{tikzcd}[column sep=30]
X \ar{rr}{p} \ar{ddd}[swap]{\inr}  && HC \ar[bend right=10, equals]{dddl} \ar{d}{Hr_{S\cap C}^{-1}} \\
&& H(S\cap C) \ar{d}{Hl_{S\cap C}} \\
&& HS \ar{d}{Hm}  \\
C+X \ar{r}{[c,p]} & HC \ar{r}{H\inl} & H(C+X)
\end{tikzcd}
\end{equation}
Therefore, there exists a unique morphism $n$ making the left-hand and upper part of the diagram below commute:
\begin{equation}\label{eq:comm-diag4}
\begin{tikzcd}[column sep=30]
X \ar[dashed]{ddr}{n} \ar{rrr}{p} \ar[bend right=10]{dddr}[swap]{\inr}  &&& HC \ar{d}{Hr_{S\cap C}^{-1}} \\
&&& H(S\cap C) \ar{d}{Hl_{S\cap C}} \\
& S \ar{rr}{s} \ar[tail]{d}[swap]{m} \pullbackangle{-45} && HS \ar{d}{Hm}  \\
&C+X \ar{r}{[c,p]} & HC \ar{r}{H\inl} & H(C+X)
\end{tikzcd}
\end{equation}
This in turn yields a unique morphism $o$ making the left-hand and upper part of the following diagram commute:
\begin{equation}
\begin{tikzcd}
X \ar[shiftarr={yshift=15}]{rr}{n} \ar[dashed]{r}{o} \ar{dr}[swap]{\id} & S\cap X \pullbackangle{-45} \ar[tail]{r}{l_{S\cap X}} \ar[tail]{d}[swap]{r_{S\cap X}} & S \ar[tail]{d}{m} \\
& X \ar{r}{\inr} & C+X 
\end{tikzcd}
\end{equation}
Thus the monomorphism $r_{S\cap X}$ is a split epimorphism (with splitting $o$), and so it is an isomorphism.
\end{proof}%

Recall from the introduction the coalgebraic formulation of the classical \Konig's lemma: In a finitely branching well-founded graph (equivalently, a well-founded $\Powf$-coalgebra), every node lies in some finite subcoalgebra. The following theorem generalizes \Konig's lemma to coalgebras over \lfp categories:

\begin{theorem}[Coalgebraic \Konig's Lemma, \injmono, \hpresint]\label{thm:koenig}
Every well-founded $H$-coalgebra is the directed join of its fg-carried well-founded subcoalgebras.
\end{theorem}
Here the join is formed in the complete lattice of subcoalgebras of the given coalgebra. As a minor subtlety, we note that subcoalgebras of well-founded coalgebras generally need not be well-founded; this would require stronger assumptions on the categorical setting such as $H$ preserving preimages rather than just binary intersection~\cite[Cor.~8.4.15]{amm25}. Therefore, it is crucial that \autoref{thm:koenig} considers \emph{well-founded} (not arbitrary) subcoalgebras.

\begin{proof}
  Let $(C,c)$ be a well-founded coalgebra and let 
  \[  m_i\colon (C_i,c_i)\monoto (C,c) \qquad (i\in I)\]
  be the family of all fg-carried well-founded subcoalgebras of $(C,c)$, viewed as a poset with respect to the usual order of subobjects. This poset is small because $\C$ is well-powered (\autoref{lem:lfp-complete-wp}). Moreover, it is directed, in fact a lattice: Given a finite subfamily $m_i\colon (C_i,c_i)\monoto (C,c)$ ($i\in I_0$), its join $m\colon \bigvee_{i\in I_0} (C_i,c_i)\monoto (C,c)$ is the image of $[m_i]_{i\in I_0}\colon \coprod_{i\in I_0} (C_i,c_i)\to (C,c)$. Note that the coalgebra $\bigvee_{i\in I_0} (C_i,c_i)$ is itself fg-carried and well-founded because the subcategory $\Coalg_{\fg,\wf}(H)\hookrightarrow \Coalg(H)$ is closed under finite coproducts (\autoref{lem:fp-fp-fin-colimits}) and strong quotients (\autoref{prop:wf}). 

Now form the join of all $m_i\colon (C_i,c_i)\monoto (C,c)$ ($i\in I$) in the complete lattice of all subcoalgebras of $(C,c)$:
\[
\begin{tikzcd}
(C_i,c_i) \ar[shiftarr={yshift=6mm}, tail]{rr}{m_i} \ar[tail]{r}{n_i} & (\bigvee_{i\in I} C_i, \ol{c}) \ar[tail]{r}{m} & (C,c)  
\end{tikzcd}
\]
This join is directed, as argued above, and formed at the level of $\C$ (\autoref{lem:colimits-coalg}). Since directed joins of subobjects in \lfp categories are directed colimits (\autoref{lem:lfp-directed-unions}), we see that $(n_i)_{i\in I}$ is a colimit cocone. We will show that $m$ is a cartesian subcoalgebra, which implies that $m$ is an isomorphism and hence proves the theorem.

To prove that $m$ is cartesian, we need to show that the following square is a pullback:
\begin{equation}\label{eq:pullback}
\begin{tikzcd}
\bigvee C_i \ar{r}{\overline{c}} \pullbackangle{-45} \ar[tail]{d}[swap]{m} & H(\bigvee C_i) \ar[tail]{d}{Hm} \\
C \ar{r}{c} & HC
\end{tikzcd}
\end{equation}

Thus suppose that $f\colon X\to H(\bigvee C_i)$ and $g\colon X\to C$ are morphisms with
\[ Hm\circ f = c\circ g.  \]
Our task is to construct the mediating morphism $k\colon X\to \bigvee C_i$. 

\begin{enumerate}
\item Assume first that $X\in \C_\fg$. Since $H$ is finitary and preserves monos, $(Hn_i)$ is a colimit cocone of monomorphisms and so $f$ factorizes as
\[ 
\begin{tikzcd}
X \ar{r}{f} \ar{dr}[swap]{p} & H(\bigvee C_i) \\
& HC_i \ar[tail]{u}[swap]{Hn_i}
\end{tikzcd}
 \]
for some $i\in I$ and $p\colon X\to HC_i$. Form the coproduct extension $(C_i+X,(c_i)_p)$ of the coalgebra $(C_i,c_i)$ by $p$ (\autoref{cons:cp}). Then
\[ [m_i,g]\colon (C_i+X,(c_i)_p)\to (C,c)\]
is a coalgebra morphism, as the diagram below commutes:
\[
\begin{tikzcd}
C_i+X \ar{rr}{(c_i)_p} \ar{dr}{[c_i,p]} \ar{dd}[swap]{[m_i,g]} && H(C_i+X) \ar{dd}{H[m_i,g]}  \\
& HC_i \ar{ur}{H\inl} \ar{dr}{Hm_i}  & \\
C \ar{rr}{c} && HC 
\end{tikzcd}
\]
For commutativity of the second component of the lower left part, we compute
\[ Hm_i\circ p = Hm\circ Hn_i \circ p = Hm\circ f = c\circ g. \]
Form the image factorization of $[m_i,g]$. Since the coalgebra $(C_i+X, (c_i)_p)$ is well-founded (\autoref{prop:wf-coalgebra-coproduct}) and $\Coalg_{\fg,\wf}(H)$ closed under strong quotients (\autoref{lem:fp-fp-fin-colimits}, \iref{prop:wf}{closed-colim-quotients}), the image of $[m_i,g]$ is an fg-carried well-founded subcoalgebra of $(C,c)$ and thus of the form $m_j\colon (C_j,c_j)\monoto (C,c)$ for some $j\in I$:
\[  
\begin{tikzcd}
(C_i+X,(c_i)_p) \ar[shiftarr={yshift=15}]{rr}{[m_i,g]} \ar[two heads]{r}{e} & (C_j,c_j) \ar[tail]{r}{m_j} & (C,c)
\end{tikzcd}
\]
We now define $k\colon X\to \bigvee_i C_i$ to be the composite
\[ k = (\, X \xto{\inr} C_i+X \xto{e} C_j \xto{n_j} \bigvee_i C_i  \,). \]
Then $k$ is the desired mediating morphism, i.e. both the left-hand and upper part of the diagram below commute:
\begin{equation}\label{eq:pb-proof}
\begin{tikzcd}
X \ar[bend left=20]{drr}{f} \ar[bend right=20]{ddr}[swap]{g} \ar{dr}{k}  && \\
& \bigvee C_i \ar{r}{\overline{c}} \pullbackangle{-45} \ar[tail]{d}[swap]{m} & H(\bigvee C_i) \ar[tail]{d}{Hm} \\
& C \ar{r}{c} & HC
\end{tikzcd}
\end{equation}
Indeed, commutativity of the left-hand part is witnessed by the following diagram:
\[
\begin{tikzcd}
X  \ar{d}[swap]{g} \ar[shiftarr={yshift=7mm}]{rrr}{k} \ar{r}{\inr} & C_i+X \ar{d}{[m_i,g]} \ar{r}{e} & C_j \ar{d}{m_j} \ar{r}{n_j} & \bigvee_i C_i \ar{d}{m} \\
C \ar[equals]{r} & C \ar[equals]{r} & C \ar[equals]{r} & C 
\end{tikzcd}
\]
Then also the upper part of \eqref{eq:pb-proof} commutes because it does so when postcomposed with the monomorphism $Hm$. The uniqueness of $k$ follows from $m$ being monic.
\item Now let $X\in \C$ be an arbitrary object. Since $\C$ is \lfp, we can express $X$ as the directed join of the family $x_j\colon X_j\monoto X$ ($j\in J$) of its finitely generated subobjects. By part (1) above, for each $j\in J$ there exists a unique morphism $k_j\colon X_j\to \bigvee_i C_i$ such that
\[ m\circ k_j = g\circ x_j. \]
The morphisms $k_j$ form a cocone over the diagram of $X_j$'s, so there exists a unique $k\colon X\to \bigvee_i C_i$ such that $k\circ x_j = k_j$ for all $j\in J$. This implies
\[ m\circ k\circ x_j = m\circ k_j = g\circ x_j  \]
and so
\[ m\circ k = g \]
because the colimit injections $x_j$ are jointly epic. Thus $k$ is the desired mediating morphism witnessing that \eqref{eq:pullback} is a pullback.\qedhere
\end{enumerate}
\end{proof}

As an immediate consequence, we obtain a useful property of the category of well-founded coalgebras, provided that finitely presentable and finitely generated objects in $\C$ coincide.

\begin{corollary}[\smash{\injmono}, \hpresint]\label{cor:wfIsLFP}
If $\C_\fp=\C_\fg$, then $\Coalg_\wf(H)$ is \lfp. 
\end{corollary}

\begin{proof}
The category $\Coalg_\wf(H)$ is cocomplete because it is closed under colimits in $\Coalg(H)$ (\iref{prop:wf}{closed-colim-quotients}). All coalgebras in $\Coalg_{\fp,\wf}(H)$ 
are finitely presentable objects of $\Coalg_\wf(H)$; this follows from the fact that all coalgebras in $\Coalg_\fp(H)$ are finitely presentable objects of $\Coalg(H)$ (\autoref{lem:fp-fg-coalg}) and that $\Coalg_\wf(H)$ is closed under (filtered) colimits in $\Coalg(H)$. By \autoref{thm:koenig}, every object of $\Coalg_\wf(H)$ is a filtered colimit of objects in $\Coalg_{\fp,\wf}(H)=\Coalg_{\fg,\wf}(H)$. Thus $\Coalg_\wf(H)$ is \lfp. 
\end{proof}

Without the restrictive assumption that $\C_\fp=\C_\fg$, we get a slightly weaker result. It makes use of the following characterization of locally presentable categories (cf.\ \autoref{rem:loc-pres}):

\begin{theorem}[{Local Generation Theorem~\cite[Thm.~1.70]{adamek_rosicky_1994}}] A category is locally presentable iff it is cocomplete and there exists a regular cardinal $\lambda$ and a set $\mathcal{A}$ of $\lambda$-generated objects such that every object is the $\lambda$-directed colimit of the diagram of its subobjects from $\mathcal{A}$. 
\end{theorem}

Note that the `if' direction of the theorem does not claim that the category is locally $\lambda$-presentable for the given regular cardinal $\lambda$. We now instantiate this theorem to the category $\Coalg_\wf(H)$, the regular cardinal $\lambda=\omega$, and the essentially small set $\mathcal{A}=\Coalg_{\fg,\wf}(H)$. The coalgebras in $\A$ are finitely generated objects of $\Coalg_\wf(H)$ because they are finitely generated in $\Coalg(H)$ (\autoref{lem:fp-fg-coalg}) and  $\Coalg_\wf(H)$ is closed under filtered colimits in $\Coalg(H)$ (\iref{prop:wf}{closed-colim-quotients}). Every coalgebra in $\Coalg_\wf(H)$ is the directed colimit of its subobjects from $\A$ (\autoref{thm:koenig}). Thus, the Local Generation Theorem yields:

\begin{corollary}[\smash{\injmono}, \hpresint]
The category $\Coalg_\wf(H)$ is locally presentable.
\end{corollary}
We conclude this section with the observation that the coalgebraic \Konig's lemma no longer holds when well-founded coalgebras are replaced with recursive ones, even in $\Set$:

\begin{example}[\Konig's Lemma Fails for Recursive Coalgebras]\label{ex:konig-lemma-recursive}
  Consider the set functor
  \[
    HX = \set{*} + \set{\,(x_1,x_2)\in X\times X \mid x_1\neq x_2\,}
  \]
  which is the quotient of $X\times X$ identifying all elements on the diagonal with a constant~$*$. Clearly, $H$ is finitary and preserves binary intersections. Let $(C,c)$ be the coalgebra with state space $C := \Z\setminus\set{0}$ and the transition structure depicted below:
  \[
  \begin{tikzcd}[row sep=3mm,column sep=4mm]
      \phantom{-}1 \arrow{r} \arrow{dr}
      & \phantom{-}2 \arrow{r} \arrow{dr}
      & \phantom{-}3 \arrow{r} \arrow{dr}
      & \cdots
      \\
      -1 \arrow{r} \arrow{ur}
      & -2 \arrow{r} \arrow{ur}
      & -3 \arrow{r} \arrow{ur}
      & \cdots
    \end{tikzcd}
     \quad\quad
  \begin{array}{l}
    c\colon C \to HC,\\
    c(k) = (-|k|-1, |k|+1).
  \end{array}
  \]
The only finite subcoalgebra of $(C,c)$ is the empty coalgebra; thus, $(C,c)$ is not the union of its finite recursive subcoalgebras. However, the coalgebra $(C,c)$ is recursive: For every $H$-algebra $(A,a)$, the constant map \[h\colon (C,c)\to (A,a) \qquad \text{defined by}\qquad h(k) :=
  a(*)\] is the unique coalgebra-to-algebra morphism. Indeed, $h$ is such a morphism because $(Hh\circ  c)(k) = *\in HA$ and so \[(a\circ Hh\circ c)(k)=a(*)=h(k)\qquad \text{for all $k\in C$}.\] To prove uniqueness of $h$, suppose that $g\colon (C,c)\to (A,a)$ is a coalgebra-to-algebra morphism, and let $k\in C$. Since $c(|k|+1) = c(-|k|-1)$,
  we necessarily have $g(|k|+1) = g(-|k|-1)$ by $g = a\circ Hg\circ c$. Thus, $Hg(c(k)) = * \in HA$ and again by $g = a\circ Hg\circ c$, we obtain
\[ g(k) = (a\circ Hg\circ c)(k) = a(*) = h(k).\]
\twnote[inline]{}
\end{example}

\section{Applications of the Coalgebraic \Konig's Lemma}\label{sec:applications}
We present a few selected instantiations of our coalgebraic \Konig's lemma (\autoref{thm:koenig}).

\subsection{Coalgebras for Set Functors}\label{sec:finitary-set}
For the category $\C=\Set$, we recover a result by Ad\'amek et al.~\cite[Prop.~9.2.19]{amm25} as a special case of \autoref{thm:koenig}:

\begin{theorem}[\Konig's Lemma for Coalgebras in $\Set$]\label{thm:koenig-coalgebra-set} Let $H$ be a finitary set functor. Then every state of a well-founded $H$-coalgebra lies in some finite subcoalgebra.
\end{theorem}

Indeed, by \iref{rem:conditions}{set-func-pres-int} we can safely assume that a finitary set functor $H$ preserves binary intersections (by modifying $H$ on the empty set). Thus \autoref{thm:koenig} applies.

\subsection{Graphs in a Topos}\label{sec:topos}
While \autoref{thm:koenig-coalgebra-set} extends \Konig's lemma from graphs to coalgebras in $\Set$, we can also take an orthogonal direction: we stick with graphs but interpret them in a \lfp \emph{elementary topos}~$\C$~\cite{Johnstone_short}. (We assume basic familiarity with topos theory in this subsection.) Examples of \lfp toposes include the topos of sets, nominal sets (see \autoref{sec:nom}), any presheaf topos $\Set^\D$, or any coherent Grothendieck topos~\cite[Rem.~D.3.3.12]{Johnstone}.  

Graphs in a topos are coalgebras for the covariant power object functor $\Pow\colon \C\to \C$~\cite[Sec.~A.2.3]{Johnstone}; the latter maps each object $C$ to its power object $\Pow C=\Omega^C$, where $\Omega$ is the subobject classifier. By cartesian closure, a coalgebra $c\colon C\to \Pow C=\Omega^C$ corresponds to a morphism from $C\times C $ to $\Omega$ and thus to a subobject of $C\times C$, that is, a \emph{relation} on $C$. The role of the finite power set functor (modelling finitely branching graphs) is now played by the finitary coreflection $\Powf$ of $\Pow$ (\autoref{lem:finitary-coreflection}). Note that $\Powf$ is generally different from the \emph{Kuratowski functor} $\mathcal{K}\colon \C\to \C$, which captures Kuratowski-finite subobjects~\cite[Thm.~9.16]{Johnstone_short}. 

To see that $\Powf$-coalgebras indeed model `finite' branching, let $i\colon \Powf\to \Pow$ denote the co-universal natural transformation witnessing that $\Powf$ is the finitary coreflection of $\Pow$. Every $\Powf$-coalgebra $c\colon C\to \Powf C$ can be viewed as a $\Pow$-coalgebra $i_C\circ c\colon C\to \Pow C$ and therefore as a
 relation $m_c\colon E_c\monoto C\times C$. Given a morphism $s\colon I\to C$ we let $m_{c,I}\colon E_{c,I} \monoto I\times C$ denote the preimage of $m_c$ under $s\times \id$:
\[
\begin{tikzcd}
E\smash{{}_{c,I}} \pullbackangle{-45} \ar{r} \ar[tail]{d}[swap]{m_{c,I}} & E\smash{{}_c} \ar[tail]{d}{m_c} \\
I\times C \ar{r}{s\times \id} & C\times C
\end{tikzcd}
\]
 We then have the following simple result. Informally, it expresses that if we consider a finitely indexed family $(s_i)_{i\in I}$ of states in $C$, then their successors can also be finitely indexed. 

\begin{lemma}\label{lem:powf-coalg-topos}
Let $\C$ be a \lfp topos, and let $c\colon C\to \Powf C$ be a coalgebra for $\Powf$. Then for every $s\colon I\to C$ ($I\in \C_\fp$) there exists $t\colon J \to C$ ($J\in \C_\fp$) such that $m_{c,I}\colon E_{c,I}\monoto I\times C$ factorizes through $\id\times t\colon I\times J\to I\times C$. 
\end{lemma}

\begin{proof}
 Recall that $\Powf C$ is the colimit of the canonical diagram $D_C\colon\C_\fp/C\to \C$ mapping $(t\colon J\to C)$ to $\Pow J$ (\autoref{lem:finitary-coreflection}). We denote the colimit injection for $t$ by $e_t\colon \Pow J \to \Powf C$.

Now let $s\colon I\to C$ where $I\in \C_\fp$. Since $(e_t)$ is a filtered colimit, the morphism $c\circ s$ factorizes through $e_t$ for some $t\colon J\to C$ ($J\in \C_\fp$), that is, there exists a morphism $I\to \Pow J$ making the left-hand part of the diagram below commute.
\[
\begin{tikzcd}
&& \Pow J \ar{dr}{\Pow t} \ar{d}{e_t} & \\
I \ar[bend left=20]{urr}{\exists} \ar{r}{s} & C \ar{r}{c} & \Powf C \ar{r}{i_C} & \Pow C
\end{tikzcd}
\]
Since the right-hand triangle also commutes, we see that the morphism $i_C\circ c \circ s$ factorizes through $\Pow t$. This morphism corresponds to the subobject $m_{c,I}\colon E_{c,I}\monoto I\times C$, so the statement of the lemma follows.
\end{proof}

Applying the coalgebraic \Konig's lemma to the case of $\Powf$-coalgebras yields:

\begin{theorem}[\Konig's Lemma for Graphs in a Topos]\label{thm:koenigTopos}
In a \lfp topos, every well-founded $\Powf$-coalgebra is the join of its fg-carried well-founded subcoalgebras.  
\end{theorem}

\begin{proof}
We verify the conditions of \autoref{thm:koenig}. 
\begin{enumerate} \item Every topos is an extensive category, so (\injmono) holds (\autoref{rem:conditions}). \item The functor $\Powf$ is finitary by \autoref{lem:finitary-coreflection}. \item The power object functor $\Pow$ preserves binary intersections~\cite[Thm.~5.9]{osius74}, and this property carries over to its finitary coreflection $\Powf$ because the latter is objectwise constructed via canonical filtered colimits (\autoref{lem:finitary-coreflection}), and filtered colimits commute with finite limits (in particular intersections) in \lfp categories~\cite[Prop.~1.59]{adamek_rosicky_1994}. \qedhere
\end{enumerate}
\end{proof}

\subsection{Binary Trees}
An important special case of \Konig's lemma, studied in constructive mathematics~\cite{di21}, is its restriction to binary trees -- commonly known as the \emph{weak \Konig's lemma}. Here a {binary} tree is one where every node has at most two children, and children are ordered, that is, binary trees form coalgebras for the polynomial functor $HX = X\times X + X + 1$ on $\Set$. For \lfp categories with nice coproducts (cf.~\autoref{rem:conditions}), the generalized \Konig's lemma also has a weak version:

\begin{theorem}[Coalgebraic Weak \Konig's Lemma]\label{cor:weak-koenigs-lemma}
Let $\C$ be a \lfp and extensive category. Then every well-founded coalgebra for the functor $HX = X\times X + X + 1$ on $\C$ is the join of its fg-carried well-founded subcoalgebras.  
\end{theorem}

\begin{proof}
Again, we only need to verify that $\C$ and $H$ satisfy the conditions of \autoref{thm:koenig}:
\begin{enumerate} \item By \autoref{rem:conditions}, every extensive category satisfies (\injmono). \item The functor $H$ is finitary because filtered colimits commute with colimits (in particular coproducts) in all categories, and with finite limits (in particular products) in \lfp categories~\cite[Prop.~1.59]{adamek_rosicky_1994}. \item The functor $H$ satisfies (\hpresint) because pullbacks commute with limits (in particular products) in all categories, and with coproducts in extensive categories~\cite[Lem.~1.4]{ls05}.
\end{enumerate}
\end{proof}

The above theorem and its proof readily extend to arbitrary polynomial functors $HX = \coprod_{i\in I} X^{n_i}$ where the index set $I$ is finite and $n_i\in \Nat$.

\subsection{Nominal Transition Systems}
\label{sec:nom}
Nominal sets~\cite{pitts13} are a categorical framework to reason about variable names and
binding (like in the $\lambda$-calculus) and have also served as a basis for the 
study of automata and transition systems over infinite alphabets, whose elements represent \emph{data values}~\cite{BojanczykEA14}. Let us recall some basic definitions.

Parametric in a fixed countably infinite set of \emph{atoms} $\At$ -- intuitively the variable
names or alphabet -- a \emph{nominal set} consists of a set $X$ with
a group action $\textnormal{`\ensuremath{\cdot}'}\colon \Perm(\At)\times X\to X$
for the group $\Perm(\At)$ of finite permutations $\pi\colon \At\to \At$, satisfying the following
\emph{finite support} property: for each $x\in X$,
there is a finite set $S\subseteq \At$ such that
$
  \text{if }\pi(a) = a \text{ for all }a\in S
  \text{ then }\pi\cdot x = x.
$

Intuitively, elements of nominal sets are syntactic objects holding atoms from $\At$,
which can be renamed using finite permutations. The
finite support property means that each element in a nominal set can
hold only finitely many atoms. For example, the set $\At^*$ of finite words over~$\At$ forms a nominal set with group action $\pi\cdot (a_1\cdots a_n) = \pi(a_1)\cdots \pi(a_n)$. Similarly, the terms of the $\lambda$-calculus (modulo $\alpha$-equivalence) with variables from $\At$ form a nominal set whose group action is capture-avoiding renaming of free variables, e.g.\ $(a\, b)\cdot (\lambda a.\,a\, b) = \lambda c.\, c\, a$.

Nominal sets form a full subcategory $\Nom$ of the category of $\Perm(\At)$-actions and equivariant maps (that is, maps preserving
the group action).
The category $\Nom$ is a \lfp boolean topos with the subobject
classifier $\Omega = 2$ being the two-element set (with trivial action).
Finitely presentable and finitely generated objects are precisely the
\emph{orbit-finite} nominal sets~\cite[Thm.~5.16]{pitts13}, that is, those nominal sets having
only finitely many elements up to renaming. For example, $\At^*$ is not orbit-finite, but the nominal set~$\At^n$ of words of fixed length $n$ is. We call a subset $M\subseteq X$ of a nominal set \emph{orbit-finite}
if $\set[\big]{\set{\pi\cdot x\mid \pi \in \Perm(\At)}\mid x\in M}$ is a finite set, and \emph{finitely supported} if it satisfies the finite support property w.r.t.\ the group action $\pi\cdot M = \set{\pi\cdot x\mid x\in M}$. The power object functor on $\Nom$ and its finitary coreflection are given by 
\begin{align*}\Pow X &= \set{M\seq X \mid \text{$M$ finitely supported}};\\
\Powf X &=\set{ M\seq X \mid \text{$M$ finitely supported and orbit-finite}}.
\end{align*} 
Applying \autoref{thm:koenig} to $\C=\Nom$ gives the following nominal version of \Konig's lemma:
\begin{theorem}[\Konig's Lemma for Coalgebras in $\Nom$]\label{thm:koenigNom}
  Let $H$ be a finitary endofunctor on $\Nom$ preserving binary intersections. Then every state of a well-founded $H$-coalgebra lies in some orbit-finite
  subcoalgebra.
\end{theorem}
\begin{proof}
  All conditions of \autoref{thm:koenig} are met because $\Nom$ is a locally finitely presentable topos, in particular extensive (cf.\ \autoref{rem:conditions})
\end{proof}
Coalgebras for finitary functors on $\Nom$ admit possibly
infinite branching as long as it is finite up to renaming (cf.\ \cite[Sec.\ 6.2]{msw16}). For illustration, let us consider coalgebras for the functor $(\Pow_\omega)^\At$ on $\Nom$. (Here $(-)^\At$ denotes the exponential.) They admit a simple concrete description: A \emph{nominal labelled transition system} (\emph{NLTS}) is given by a nominal set $C$ of \emph{states} and family of relations $\mathord{\xto{a}} \seq C\times C$  ($a\in \At$) such that \[x\xto{a} y \qquad\text{implies}\qquad \pi\cdot x \xto{\pi(a)} \pi\cdot y \qquad \text{for all $\pi \in \Perm(\At)$}.\] It is \emph{image-orbit-finite} if for all $x\in C$ and $a\in \At$ the set $\set{y\in C\mid x\xto{a} y}$ is orbit-finite.

An NLTS can be viewed as an ordinary LTS by forgetting the group action.

\begin{lemma}\label{lem:powfnom}
  Coalgebras for $\smash{(\Pow_\omega)^\At}$ are in bijective correspondence with image-orbit-finite NLTS. Moreover, a coalgebra is well-founded iff the
  corresponding LTS has no infinite paths.
\end{lemma}

\begin{proof}
Since $\Pow X = 2^X$, we have for all nominal sets $C$ and $X$  the bijective correspondence
  \[
    (C\to \Pow X)
    \Leftrightarrow
    (C\to 2^X)
    \Leftrightarrow
    (C\times X\to 2) \Leftrightarrow \text{equivariant subsets of $C\times X$},
  \]
  where a subset of nominal set is \emph{equivariant} if it is closed under the group action. Thus a coalgebra $c\colon C\to (\Pow C)^\At \cong (2^C)^\At \cong 2^{\At\times C}$ corresponds to an equivariant subset of $C\times \At\times C$. This means precisely that the corresponding family of relations $\xto{a}\,\seq C\times C$ ($a\in \At$) satisfies the defining property of an NLTS. Moreover, $c$ factorizes through $(\Powf C)^\At$ iff for each $x\in C$ and $a\in \At$ the set $\set{y\in C\mid x\xto{a} y}$ is orbit-finite (by definition of $\Powf$). This means precisely that the corresponding NLTS is image-orbit-finite.

Regarding the well-foundedness statement, we can reuse the argument of \autoref{ex:wfCoalg}. We only need to note that given a coalgebra $c\colon C\to (\Powf C)^\At$, the set $S\seq C$ of all states that do not lie on any infinite path forms an equivariant subset of $C$. Since complements of equivariant subsets are equivariant, it suffices to show that $C\setminus S$ is equivariant. But this is easy to see: if $x\in C\setminus S$, there exists an infinite path $x\xto{a_1} x_1 \xto{a_2} x_2 \xto{a_3}\cdots$ starting in $x$. By equivariance of transitions, we get an infinite path $\pi\cdot x\xto{\pi(a_1)} \pi\cdot x_1 \xto{\pi(a_2)} \pi\cdot x_2 \xto{\pi(a_3)}\cdots$ for all $\pi\in\Perm(\At)$, and so $\pi\cdot x\in C\setminus S$.
\end{proof}
Note that an image-orbit-finite NLTS is generally not finitely branching as an LTS, both due to the infinity of the label set $\At$ and due to the fact that each state $x$ has a possibly infinite set of $a$-successors for a given label $a$. Thus, the classical \Konig's lemma does not apply here. However, the orbit-finite branching and equivariance of transitions is enough to get following result, which is an instance of \autoref{thm:koenigNom}:

\begin{theorem}[\Konig's Lemma for Nominal LTS]\label{thm:koenigNomGraph}
  Every state of a well-founded image-orbit-finite NLTS lies in some orbit-finite subcoalgebra.
\end{theorem}
\begin{proof}
We verify the conditions of  \autoref{thm:koenigNom}. Since the functor $(\Powf)^\At$ is the composite of $\Powf$ and $(-)^\At$, it suffices to show that both functors are finitary and preserve binary intersections.
\begin{enumerate} \item 
The functor $\Powf$ is finitary by \autoref{lem:finitary-coreflection}.
\item The functor $\Powf$ preserves binary intersections as shown in the proof of \autoref{thm:koenigTopos}. 
\item To see that $(-)^\At$ is finitary, we use the criterion of \autoref{lem:finitary-criterion}. We need to show that every orbit-finite equivariant subset $S\seq X^\At$ lies in $T^\At$ for some orbit-finite equivariant subset $T\seq X$. This is equivalent to saying that the corresponding equivariant map $f\colon S\times \At \to X$ factorizes through some orbit-finite equivariant subset $T\subto X$, i.e.\ its image is orbit finite. But this is clear because $S\times \At$ is orbit-finite (using that orbit-finiteness is preserved by finite products~\cite{pitts13}) and orbits are mapped to orbits by equivariant maps. 
\item The functor $(-)^\At$ preserves intersections, in fact all limits, as it is a right adjoint. \qedhere  
\end{enumerate} 
\end{proof}

\subsection{Convex Transition Systems}
\label{sec:convex}
The category of {convex sets}~\cite{stone-49} provides a natural setting for studying probabilistic automata and transition systems, in particular systems that combine probabilistic choice with non-determinism~\cite{bss17,mv20,msv21,bsv22}. Sets of probability distributions are modelled as objects in the category, while non-determinism is captured by a convex version of the power set functor. 

A \emph{convex set} is a set \(X\) equipped with a family of binary operations \( +_{r}  \colon X \times X \rightarrow X\) (\(r \in [0, 1]\)) subject to the following equations, where \(s' = r + s -rs \ne 0\) and \(r' = \frac{r}{s'}\): 
  \begin{equation}\label{eq:convex-set-ax} x +_{r} x = x,
    \quad
    x +_{0} y = y,
    \quad
    x +_{r} y = y +_{1-r} x,
    \quad
    x +_{r} (y +_{s} z) = (x +_{r'} y ) +_{s'} z.
  \end{equation}
Convex sets are an algebraic model of structures that admit convex combinations of elements. The prototypical example of convex sets are convex subsets $X\subseteq \mathbb{R}^n$ with the operations \[\vec x +_r \vec y := r\cdot \vec x
+ (1-r)\cdot \vec y.\] Another example is the set $\Dist X$ of finitely supported probability distributions on a set $X$ (i.e.\ maps $\varphi\colon X\to [0,1]$ such $\varphi(x)=0$ for all but finitely many $x\in X$ and $\sum_{x\in X} \varphi(x) = 1$) with convex structure given by \[\varphi+_r \psi = \big(x \mapsto r \cdot \varphi(x) + (1-r) \cdot \psi(x)\big).\]

 A map \(f \colon X \rightarrow Y\) between convex sets  %
  is \emph{affine} if \[f(x +_{r} x') = f(x) +_{r} f(x') \qquad \text{for \(x, x' \in X\) and \(r \in [0, 1]\)}.\] We let $\Conv$ denote the category of convex sets and affine maps. It is \lfp (as is every algebraic category given by a finitary equational theory~\cite{adamek_rosicky_1994}), and finitely presentable convex sets coincide with finitely generated ones~\cite[Cor.~5.5]{sw15}. Note that in algebraic categories, finitely generated objects are the algebras with a finite set of generators~\cite[Ex.~1.68]{adamek_rosicky_1994}. Thus a convex set $X$ is finitely generated iff there exists a finite subset $X_0\seq X$ such that every $x\in X$ is an (iterated) convex combination of elements of $X_0$. 

The convex counterpart of the finite power set functor $\Pow_\omega$ is the \emph{convex power set functor}~$\CPowf$ on $\Conv$. It maps a convex set $X$ to the convex set $\CPowf X$ of all finitely generated (but not necessarily finite) convex subsets of~$X$. The operation $+_r$ on $\CPowf X$ is given by \[S+_r T = \set{x+_r y \mid x\in S,\, y\in T} \qquad \text{for $r\in (0,1)$},\] and determined by the second and third axiom of \eqref{eq:convex-set-ax} for $r\in \{0,1\}$, e.g.\ \[\emptyset +_0 S = S = S+_1 \emptyset.\] The functor~$\CPowf$ is the composite of the \emph{non-empty} convex power set functor and the \emph{black-hole extension} functor $(-)+1$ from~\cite[Sec.~5]{bss17}.

A coalgebra $c\colon C\to \CPowf C$ is an \emph{fg-branching convex graph}, that is, a graph where \begin{enumerate} \item  the set $C$ of nodes carries the structure of a convex set; \item for every $x\in C$, the set $c(x)$ of successors forms a finitely generated convex subset of $C$; 
  \item  $c(x+_r y)=c(x)+_r c(y)$ for all $x,y\in C$ and $r\in [0,1]$. 
\end{enumerate} An fg-branching convex graph can be viewed as an ordinary graph by forgetting the convex structure on $C$. Well-foundedness is then the expected concept:

\begin{lemma}\label{lem:convex-well-founded}
A coalgebra $c\colon C\to \CPowf C$ is well-founded iff it has no infinite paths.
\end{lemma}

\begin{proof}
The argument is analogous to the case of graphs (\autoref{ex:wfCoalg}). We only need to observe that for every coalgebra $c\colon C\to \CPowf C$, the set $S\seq C$ of all nodes that lie on no infinite path forms a convex subset of $C$. To see this, let $x,y\in S$ and $z=x+_r y$ for some $r\in [0,1]$. We need to prove $z\in S$. If $r\in \{0,1\}$ then $z\in \{x,y\}$ and so $z\in S$. Now let $r\in (0,1)$, and suppose towards a contradiction that $z\not\in S$, that is, there is an infinite path $z\to z_1 \to z_2\to\cdots$ starting with $z$. Then $z_1\in c(z)=c(x+_r y) = c(x)+_r c(y)$. This means that there exist $x\to x_1$ and $y\to y_1$ such $z_1=x_1+_r y_1$. Repeating this argument with $x_1,y_1,z_1$ instead of $x,y,z$ yields $x_1\to x_2$ and $y_1\to y_2$ such that $z_2=x_2+_r y_2$ etc. In this way, we obtain an infinite path $x\to x_1\to x_2\to \cdots$, in contradiction to $x\in S$.    
\end{proof}

When regarded as an ordinary graph, an fg-branching convex graph is generally not finitely branching. However, the fg-branching property is enough to get a version of \Konig's lemma. Indeed, applying \autoref{thm:koenig} to the category $\Conv$ and the functor $\CPowf$ yields:

\begin{theorem}[\Konig's Lemma for Convex Graphs]\label{thm:koenig-convex}
Every node of a well-founded fg-branching convex graph lies in some fg-carried $\CPowf$-subcoalgebra.
\end{theorem}

\begin{proof}
As usual, we only need to show that the conditions of \autoref{thm:koenig} are satisfied: \begin{enumerate}\item Condition (\injmono) follows from the description of coproducts in $\Conv$ given in~\cite[Prop.~5]{jww15}. \item  Clearly, the functor $\CPowf$ preserves binary intersections. (Note that since $\Conv$ is an algebraic category, limits in $\Conv$ are formed like in $\Set$.) \item To check that $\CPowf$ is finitary, we use the criterion of \autoref{lem:finitary-criterion}: we need to show that every finitely generated convex set $\mathcal{S}\seq \CPowf C$ is a subset of $\CPowf M$ for some finitely generated convex set $M\seq C$. Let $\S_0$ be a finite set of generators of $\S$. Then each $S\in \S_0$ is a finitely generated convex subset of $C$ and thus has a finite set of generators $S_0\seq S$. Let $M\seq C$ be the convex subset generated by the finite set $\bigcup_{S\in \S_0} S_0$. Then $S\seq M$ for every $S\in \S_0$, and so $S\seq M$ for every $S\in \S$ because $\S$ is generated by $\S_0$.
\qedhere
\end{enumerate}
\end{proof}

\section{Initial Algebras From Finite Well-Founded and Recursive Coalgebras}\label{sec:initial-algebras}
In this section we establish our second fundamental result about well-founded coalgebras: the initial algebra for $H$ can be constructed as the colimit of all fp-carried well-founded coalgebras. Note that, in contrast, \autoref{thm:koenig} involves \emph{fg-carried} well-founded coalgebras. The proof is another intricate application of coproduct extension (\autoref{cons:cp}). We present our result in a slightly more general form that applies uniformly to the related cases of well-founded and recursive coalgebras.

\subsection{A Fixed Point Construction}

The key to our initial algebra construction is a general construction of a fixed point of the functor $H$ as a colimit of coalgebras. In the following, we say that a full subcategory $\D$ of $\Coalg(H)$ is \emph{closed under coproduct extensions by finitely presentable objects} if for every coalgebra $(C,c)\in \D$ and every morphism $p\colon X\to HC$ of $\C$ where $X\in \C_\fp$, one has $(C+X,c_p)\in \D$.

\begin{theorem}[Fixed Point Theorem]\label{thm:fixed-point}
Suppose that $D\colon \D\subto \Coalg(H)$ is a full subcategory that is essentially small, filtered, and closed under coproduct extensions by finitely presentable objects. Then the coalgebra
\[ (T,t) = \colim D  \]
is a fixed point of $H$, that is, the coalgebra structure $t\colon T\to HT$ is an isomorphism.
\end{theorem}

\takeout{}
\begin{proof} %
We denote the colimit injections by
\[ c^\#\colon (C,c)\to (T,t) \qquad \text{where} \qquad (C,c)\in \D. \]
Since the forgetful functor $U\colon \Coalg(H)\to \C$ creates colimits (\autoref{lem:colimits-coalg}), we know that $(c^\#)$ is a colimit cocone over $UD$. Moreover, $(Hc^\#\circ c)$ is a cocone over $UD$ and $t$ is the unique mediating morphism such that
\[  
\begin{tikzcd}
T \ar[dashed]{r}{t} & HT \\
C \ar{u}{c^\#} \ar{r}{c} & HC \ar{u}[swap]{Hc^\#}
\end{tikzcd}
\]
commutes for every $(C,c)\in \D$. It suffices to show that $(Hc^\#\circ c)$ forms a colimit cocone over $UD$; then $t$ is an isomorphism by the uniqueness of colimits. 
Since the diagram scheme $\D$ is filtered by assumption, we use the criterion of \autoref{lem:filtered-colimits}, i.e.\ we show that for $X\in \C_\fp$, every morphism from $X$ to $HT$ factorizes essentially uniquely through the cocone $(Hc^\# \circ c)$.
\begin{enumerate}
\item Let $q\colon X\to HT$ be a morphism. Since $H$ is finitary, $(Hc^\#)$ is a colimit cocone over $HUD$. Since $X$ is finitely presentable, there exists $(C,c)\in \D$  and $p\colon X\to HC$ such that
\[ q=Hc^\# \circ p. \]
Consider the coproduct extension $(C+X,c_p)$ of $(C,c)$ by $p$ (\autoref{cons:cp}). By our assumption that $\D$ is closed under coproduct extensions by finitely presentable objects, we have $(C+X,c_p)\in \D$. Moreover, the following diagram commutes:
\[
\begin{tikzcd}
X \ar{dd}[swap]{\inr} \ar{rr}{q} \ar{dr}{p} & & HT \\
& HC \ar{ur}{Hc^\#} \ar{dr}{H\inl} & \\
C+X \ar{ur}{[c,p]} \ar{rr}{c_p} & & H(C+X) \ar{uu}[swap]{Hc_p^\#}
\end{tikzcd}
\]
Indeed, the right-hand triangle commutes because $\inl\colon (C,c)\to (C+X,c_p)$ is a coalgebra morphism (\autoref{lem:coproduct-inl-coalg-mor}) and $(-)^\#$ is a cocone. The remaining parts clearly commute. Hence the outside commutes, which proves that $q$ factorizes through $Hc_p^\# \circ c_p$ via $\inr$.
\item To prove essential uniqueness of factorizations, let $(C,c)\in \D$ and let $f,f'\colon X\to C$ be morphisms such that 
\[ q:= Hc^\# \circ c \circ f = Hc^\# \circ c \circ f'.\]
Our task is to produce a coalgebra morphism in $\D$ that merges $f$ and $f'$. Put
\[ p:=c \circ f\qqand p':=c\circ f'.  \]
 Since $(Hc^\#)$ is a colimit cocone over $HUD$ and $X$ is finitely presentable, there exists $g\colon (C,c)\to (\ol{C},\ol{c})$ in $\D$ such that  
\[ Hg\circ p = Hg\circ p'.\]
Therefore, after replacing $f,f'$ with $g\circ f, g\circ f'$, we may assume w.l.o.g.\ that $p=p'$.
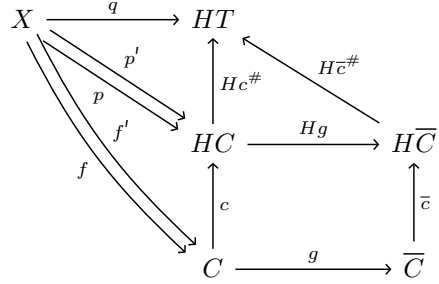
\[
\begin{tikzcd}[row sep=3em, column sep=5em]
X \ar{r}{q} \ar[shift right=1]{dr}[swap]{p} \ar[shift left=1]{dr}{p'} \ar[shift right=1, bend right=10]{ddr}[swap]{f} \ar[shift left=1, bend right=10]{ddr}{f'} & HT & \\
& HC \ar{r}{Hg} \ar{u}[swap]{Hc^\#} & H\ol{C} \ar{ul}[swap]{H\ol{c}^\#} \\
& C \ar{u}[swap]{c} \ar{r}{g} & \ol{C} \ar{u}[swap]{\ol{c}}
\end{tikzcd}
\]
Then $[\id,f],[\id,f']\colon (C+X,c_p)\to (C,c)$ are coalgebra morphisms; this is shown by the commutative diagram below for $[\id,f]$, and analogously for $[\id,f']$, using that $p=p'$.
\[
\begin{tikzcd}[column sep=35]
C+X \ar{r}{[\id,f]}  \ar{d}[swap]{[c,p]} \ar[shiftarr={xshift=-14mm}]{dd}[swap]{c_p} & C \ar{dd}{c} \\
HC \ar[bend left=10]{dr}{\id}\ar{d}[swap]{H\inl} & \\
H(C+X) \ar{r}{H[\id,f]} & HC
\end{tikzcd}
\]
Since the category $\D$ is filtered by assumption, there exists a morphism  $h\colon (C,c)\to (C',c')$ in $\D$ merging $[\id,f]$ and $[\id,f']$:
\[  
\begin{tikzcd}
(C+X,c_p) \ar[shift left=1]{r}{[\id,f]} \ar[shift right=1]{r}[swap]{[\id,f']} & (C,c) \ar{r}{h} & (C',c').
\end{tikzcd}
\]
In particular, we have $h\circ f= h\circ f'$, as required.\qedhere
\end{enumerate}
\end{proof}

\begin{example}\label{exa:ratfp}
\begin{enumerate}
\item For $\D=\Coalg_\fp(H)$, the coalgebra $(T,t)$ is known as the \emph{rational fixed point} of $H$. The algebra $(T,t^{-1})$ is characterized as the initial \emph{iterative} algebra~\cite{amv06}.
\item The choice $\D=\Coalg_\fg(H)$ yields the \emph{locally finite fixed point} of $H$, corresponding to the initial \emph{fg-iterative} algebra~\cite{mpw16}.
\end{enumerate}
\end{example}

For the case where $\D$ is a full subcategory of $\Coalg_\rec(H)$, one obtains an initial algebra: 
\begin{theorem}[Abstract Initial Algebra Theorem]\label{thm:abstract-initial-algebra}
In the setting of \autoref{thm:fixed-point}, suppose that all coalgebras in $\D$ are recursive. Then $(T,t^{-1})$ is an initial algebra for $H$.
\end{theorem}

\begin{proof}
If all coalgebras in $\D$ are recursive, then $(T,t)$ is recursive by \autoref{lem:rec-coalg-colimits}. Since algebra morphisms $(T,t^{-1})\to (A,a)$ correspond precisely to coalgebra-to-algebra morphisms $(T,t)\to (A,a)$, we conclude that the algebra $(T,t^{-1})$ is initial. 
\end{proof}

\subsection{Initial Algebras from Well-Founded Coalgebras}
As the main instance of the Abstract Initial Algebra Theorem, we consider the cases $\D=\fwf, \Coalg_{\fg,\wf}(H)$:

\begin{theorem}[First Initial Algebra Theorem,\,\smash{\injmono},\,\hpresint]\label{thm:wf-initial-algebra}
The initial algebra for $H$ is both
\begin{enumerate}
\item  the colimit of all fp-carried well-founded coalgebras, and
\item the colimit of all fg-carried well-founded coalgebras.
\end{enumerate}
\end{theorem}

\begin{proof}
We instantiate \autoref{thm:abstract-initial-algebra} to $\D=\fwf$. We verify that $\D$ satisfies the required conditions:
\begin{enumerate}
  \item The category $\fwf$ is filtered because it is finitely cocomplete; in fact, it closed under finite colimits in $\Coalg(H)$ because both $\C_\fp\hookrightarrow \C$ and $\Coalg_\wf(H)\hookrightarrow \Coalg(H)$ are closed under finite colimits (\autoref{lem:fp-fp-fin-colimits}, \autoref{prop:wf}).
  \item The category $\fwf$ is closed under coproduct extensions by finitely presentable objects by \autoref{prop:wf-coalgebra-coproduct} and because $\C_\fp\subto \C$ is closed under finite coproducts. 
  \item All coalgebras in $\D$ are recursive by \iref{prop:wf}{wf-to-rec}. 
\end{enumerate}
The reasoning for $\D=\Coalg_{\fg,\wf}(H)$ is analogous.
\end{proof}

 Under the extra assumption $\C_\fp=\C_\fg$, one can alternatively derive \autoref{thm:wf-initial-algebra} directly from the coalgebraic \Konig's lemma rather than the Abstract Initial Algebra Theorem. This uses the following known relation between well-founded coalgebras and initial algebras:

\begin{notheorembrackets}
\begin{theorem}[{\cite[Cor.~8.5.8]{amm25}}]\label{thm:final-well-founded-coalgebra} 
If $(T,t)$ is a terminal well-founded coalgebra (i.e.\ a terminal object in $\Coalg_\wf(H)$), then $t$ is an isomorphism and $(T,t^{-1})$ is an initial algebra.
\end{theorem}
\end{notheorembrackets}

\begin{proof}[Alternative proof of \autoref{thm:wf-initial-algebra} for $\C_\fp=\C_\fg$] 
In every locally finitely presentable category~$\mathbb{A}$ with generating set $\A$ of finitely presentable objects, the colimit of the full subcategory $\A\hookrightarrow \mathbb{A}$ is the terminal object of $\mathbb{A}$. Applying this to $\mathbb{A}=\Coalg_\wf(H)$ and $\A=\Coalg_{\fp,\wf}(H)$ (\autoref{cor:wfIsLFP}) shows that the colimit of all fp-carried well-founded coalgebras is the terminal well-founded coalgebra, hence the initial algebra by \autoref{thm:final-well-founded-coalgebra}. 
\end{proof}

\subsection{Initial Algebras from Recursive Coalgebras}\label{sec:second-init-alg}
Wißmann and Milius~\cite{wm24} have shown an analogue of \autoref{thm:wf-initial-algebra} for \emph{recursive} instead of well-founded coalgebras. This is another instance of \autoref{thm:abstract-initial-algebra}, this time applied to $\D=\Coalg_{\fp,\rec}(H)$ resp.\ $\Coalg_{\fg,\rec}(H)$:

\begin{theorem}[Second Initial Algebra Theorem~\cite{wm24}]\label{thm:rec-initial-algebras}
The initial algebra for $H$ is both
\begin{enumerate} \item the colimit of all fp-carried recursive coalgebras, and
  \item the colimit of all fg-carried recursive coalgebras.
\end{enumerate}
\end{theorem}

\begin{proof}
We instantiate \autoref{thm:abstract-initial-algebra} to $\D=\Coalg_{\fp,\rec}(H)$. Again, we only need to verify that $\D$ satisfies the required conditions:
\begin{enumerate}  
  \item The category $\Coalg_{\fp,\rec}(H)$ is filtered because it is finitely cocomplete; in fact, it is closed under finite colimits in $\Coalg(H)$ because both $\C_\fp\hookrightarrow \C$ and $\Coalg_{\rec}(H)\hookrightarrow \Coalg(H)$ are closed under finite colimits (\autoref{lem:fp-fp-fin-colimits}, \autoref{lem:rec-coalg-colimits}).
  \item Closure of $\Coalg_{\fp,\rec}(H)$ under coproduct extensions by finitely presentable objects follows from $\C_\fp\subto \C$ being closed under coproducts and \autoref{prop:rec-coalgebra-coproduct} below.
  \item All coalgebras in $\D$ are recursive by definition. 
\end{enumerate}
The reasoning for $\D=\Coalg_{\fg,\rec}(H)$ is analogous.
\end{proof}
 
\begin{proposition}\label{prop:rec-coalgebra-coproduct}
For every recursive coalgebra $(C,c)$ and every morphism $p\colon X\to HC$, the coproduct extension $(C+X,c_p)$ is a recursive coalgebra.
\end{proposition}

\begin{proof}
Suppose that $(C,c)$ is a recursive coalgebra and $p\colon X\to HC$ is a morphism. We show that $(C+X,c_p)$ is a recursive coalgebra: for every $H$-algebra $(A,a)$, there exists a unique coalgebra-to-algebra morphism from $(C+X,c_p)$ to $(A,a)$.

\medskip\noindent\emph{Existence.} Since $(C,c)$ is recursive, there exists a unique coalgebra-to-algebra morphism $h\colon (C,c)\to (A,a)$. We claim that the composite
\[ C+X \xto{\id+p} C+HC \xto{h+Hh} A+HA \xto{[\id,a]} A \]
is a coalgebra-to-algebra morphism from $(C+X, c_p)$ to $(A,a)$, which means that the outside of the diagram below commutes:
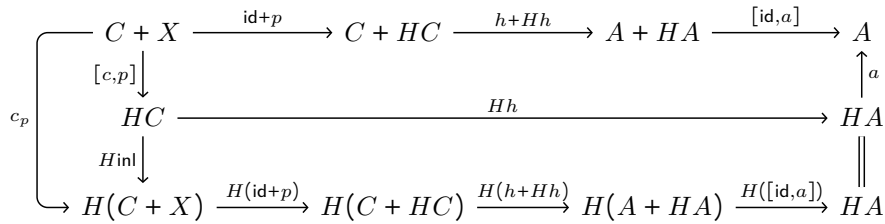
\[
\begin{tikzcd}[column sep=35]
C+X \ar[shiftarr={xshift=-14mm}]{dd}[swap]{c_p} \ar{r}{\id+p} \ar{d}[swap]{[c,p]} & C+HC \ar{r}{h+Hh} & A+HA \ar{r}{[\id,a]} & A \\
HC \ar{rrr}{Hh} \ar{d}[swap]{H\inl} & & & HA \ar{u}[swap]{a} \\
H(C+X) \ar{r}{H(\id+p)} & H(C+HC) \ar{r}{H(h+Hh)} & H(A+HA) \ar{r}{H([\id,a])} & HA \ar[equals]{u}
\end{tikzcd}
\]
Indeed, the lower part commutes trivially, and the left-hand part by definition of $c_p$. For the upper part consider its precomposition with the coproduct injections of $C+X$. The precomposition with $\inl$ yields the diagram expressing that $h$ is a coalgebra-to-algebra morphism, and the precomposition with $\inr$ commutes trivially.

\medskip\noindent\emph{Uniqueness.} Let $[h,k]\colon (C+X,c_p)\to (A,a)$ be a coalgebra-to-algebra morphism. Then the outside of the diagram
\[
\begin{tikzcd}
C+X \ar[shiftarr={xshift=-18mm}]{dd}[swap]{c_p} \ar{r}{[h,k]} \ar{d}[swap]{[c,p]} & A \\
HC \ar{r}{Hh} \ar{d}[swap]{H\inl} & HA \ar{u}[swap]{a} \\
H(C+X) \ar{r}{H[h,k]} & HA \ar[equals]{u}
\end{tikzcd}
\]
commutes, as do its lower and left-hand parts. Hence the upper part commutes. Precomposing it with $\inl$ shows that $h\c (C,c)\to (A,a)$ is a coalgebra-to-algebra morphism; thus $h$ is uniquely determined because $(C,c)$ is recursive. Precomposition with $\inr$ shows that $k=a\circ Hh\circ p$, so $k$ is uniquely determined by $h$.
\end{proof}

We note that \autoref{thm:rec-initial-algebras}, unlike \autoref{thm:wf-initial-algebra}, does not impose any conditions on~$\C$ and~$H$ beyond those of \autoref{asm}. However, in the case where the additional conditions (\injmono) and (\hpresint) are satisfied, the first initial algebra theorem can be regarded as an improvement of the second one since well-foundedness of a coalgebra is in practice often easier to verify than recursivity, and moreover, the initial algebra construction involves a smaller diagram. Indeed, under these conditions well-founded coalgebras form a subclass of recursive ones (\autoref{prop:wf}), and this subclass is generally proper, as witnessed by the coalgebra in \autoref{ex:konig-lemma-recursive}.

\begin{rem}
It is instructive to compare \autoref{thm:wf-initial-algebra} and \ref{thm:rec-initial-algebras} with the well-known iterative construction of initial algebras due to  Ad\'amek~\cite{adamek74}:
For every finitary endofunctor~$H$ on a category $\C$ with an initial object $0$ and colimits of $\omega$-chains, the initial algebra is the colimit of the $\omega$-chain
\[
(0,\initial)\xto{\smash{\initial}}
(H0,H\initial)\xto{\smash{H\initial}}
(HH0,HH\initial) \xto{\smash{HH\initial}}
\cdots \xto{\smash{H^{n-1}\initial}} 
(H^n0,H^n\initial)\xto{\smash{H^n\initial}} \cdots 
\]
in $\Coalg(H)$, where $\initial\colon 0\to H0$ is the unique morphism. Similar to \autoref{thm:wf-initial-algebra} and \ref{thm:rec-initial-algebras}, this construction approximates the initial algebra from below: intuitively, one thinks of the initial algebra as the set of all finite trees with branching type $H$, and of $H^n 0$ as the set of trees of height less than $n$. The approximating coalgebras $(H^n0,H^n\initial)$ are recursive~\cite[Prop.~6]{cuv06}, but in contrast to the constructions underlying \autoref{thm:wf-initial-algebra} and \ref{thm:rec-initial-algebras}, they are generally neither well-founded nor fp-carried.
\end{rem}

\subsection{Agda Formalization}
In the previous work by Wißmann and Milius~\cite{wm24},
the construction of the initial algebra via fp-carried recursive coalgebras
(\autoref{sec:second-init-alg}) has also been implemented in the Agda proof
assistant and the Agda categories library~\cite{agda-categories}. Since our new fixed point theorem
(\autoref{thm:fixed-point}) simplifies and generalizes the original proof of
\cite{wm24}, we have formalized \autoref{thm:fixed-point} in Agda, together
with a couple of the instances discussed above. Our formalization is available on
\begin{center}
  \url{\agdaurl index.html}
\end{center}
\paragraph*{Formalization Challenges.}
Like in the previous formalization~\cite{wm24}, a major challenge in the
formalization of \autoref{thm:fixed-point} is to ensure
that the diagrams are small, such that their colimits exist. In  \lfp categories,
the finitely presentable objects form up to
isomorphism a set; e.g.~all finite sets form a proper class but form only a set
up to bijection. In order to define a small diagram of all coalgebras with a finitely presentable carrier, one needs to find a representative per isomorphism class of finitely presentable objects.
In classical reasoning, we can obtain this representative by choice principles, but
in the constructive setting of Agda, we construct it from the generating set $\A$ of a \lfp category $\C$ (cf.~\autoref{sec:lfp}).
In general, not every finitely presentable object of $\C$ is isomorphic to some object from the
set $\A$, but every finitely presentable object is a split-quotient of some object in
$\A$~\cite[Def.\ 1.9, Remark]{adamek_rosicky_1994}. Thus, we consider the
following index set:
\[
  \bar\A :=
  \set{ (P,e) \mid P \in \A, e\colon P\to P }
\]
and by cocompleteness of $\C$, we obtain:
\begin{prop}
  An object $X$ of $\C$ is finitely presentable iff $X$ is the coequalizer of $\id_P$ and $e$ for
  some $(P,e)\in \bar \A$.
\end{prop}
\begin{proof}
  For a finitely presentable object $X$, there is some $P\in \A$ and
  a split quotient $q\colon P\to X$, that is,
  there exists $s\colon X\to P$ with $q\circ s = \id_X$. Then, it is not hard to see that
  \[
    \begin{tikzcd}
      P
      \arrow[shift right=1]{r}[swap]{\id_P}
      \arrow[shift left=1]{r}{s\circ q}
      & P
      \arrow{r}{q}
      & X
    \end{tikzcd}
  \]
  is a coequalizer. Conversely, for every $P\in \A$ and $e\colon P\to P$, the
  coequalizer of $e$ and $\id_P$ is finitely presentable, because finitely presentable objects are closed under
  finite colimits.
\end{proof}
In the formalization, the objects of our diagram scheme are coalgebras
whose carrier is represented by an element of $\bar \A$ and who satisfy a property parameter (which can be instantiated to recursiveness).
The previous formalization \cite{wm24} considered
recursive coalgebras whose carrier is in $\A$. Using $\bar \A$ instead makes it easier to
verify the assumptions of the new fixed point theorem
(\autoref{thm:fixed-point}), concretely the verification of filteredness and
closure under coproduct extension.

\paragraph*{File Structure.}
The new proofs took additional 2000 lines of code on top of the
previous formalization \cite{wm24}. We mention important parts:
\begin{itemize}
\item \autoref{lem:filtered-colimits} (\agdafileref{Filtered-Cocone-Limitting},
185 LOC) on the characterization of filtered colimits is implemented in full
generality and is the driving motor of the simplified proof of
\autoref{thm:fixed-point}.
\item The coproduct extension \autoref{cons:cp} and related lemmas
(\agdafileref{Coproduct-Extension}, 83 LOC) are adapted straightforwardly
from \cite{wm24}.
\item The formalization of fixed point theorem
  is slightly less general than \autoref{thm:fixed-point}, because 
  the formalization works with a filtered diagram of fp-carried coalgebras
  \\
  (\agdafileref{Filtered-Coalgebra-Fixpoint}, 400 LOC).
  Still, it uses \autoref{lem:filtered-colimits} and so is shorter than the
  previous proofs specific to coalgebra morphisms between recursive coalgebras~\cite{wm24}.
  In contrast, \autoref{thm:fixed-point} requires the diagram scheme to be filtered, which was implicit in the previous specific proof for recursive coalgebras. Thus,
  to instantiate the theorem, we had to formalize the proof of filteredness in the instances
  (\agdafileref{Coalgebra.IdxProp-Filtered} 230 LOC), which also required
  additional administrative overhead and the formalization of categorical facts:
  \begin{itemize}
  \item Interaction of colimits with Agda's levels: \agdafileref{Colimit-Lift} (143 LOC)
  \item The characterization of coequalizers via colimits of diagrams:\\
  \agdafileref{Category.ParallelArrow.Colimit} (180 LOC)
  \item The proof that finitely presentable objects are closed under coequalizers:\\ \agdafileref{Presentable-Coequalizer} (120 LOC)
  \end{itemize}
\item The construction and proof of the index set $\bar \A$ is formalized
  in \agdafileref{Accessible-fp-family} (100 LOC).
\item \agdafileref{Fixpoint-Instances} (170 LOC) constructs the following fixed points:
  \begin{itemize}
  \item The colimit of all fp-carried coalgebras: this is the first
    formalization of the rational fixed point (\autoref{exa:ratfp}).
  \item The colimit of all recursive, fp-carried coalgebras, recovering \cite{wm24}.
  \end{itemize}
\end{itemize}

\section{Conclusion and Future Work}
We have established two novel results in the theory of well-founded coalgebras over locally finitely presentable categories. Our first contribution is a coalgebraic version of \Konig's lemma, lifting a previous result for coalgebras over sets~\cite{amm25} to full categorical generality. We have illustrated the scope of our coalgebraic \Konig's lemma by exploring systems with state spaces beyond sets, such as graphs in a topos, and nominal and convex transition systems. As our second contribution, we devised a new construction of the initial algebra for a finitary functor as the colimit of all well-founded coalgebras with finitely presentable state space. This result emerges from a general fixed point construction that also captures the initial algebra construction from fp-carried recursive coalgebras introduced in earlier work~\cite{wm24}, and yields a short and transparent new proof for the latter.

One key application of well-foundedness is termination analysis of programs~\cite{cpr11}, which amounts to associating cleverly chosen well-founded relations to the state graph of a program. We aim to explore a coalgebraic perspective on termination analysis techniques, in particular in relation with the recent advances in modelling stateful languages as coalgebras~\cite{gmstu25}. First steps in this direction have been made in our recent work with Alexandru~\cite{auw26}.

A further interesting route is studying connections between well-founded coalgebras and \emph{well-structured transition systems}~\cite{fs01}, which involve a well-quasi-order on states and provide a uniform framework for a class of decidability results in verification, e.g.~for Petri nets.

\paragraph*{Acknowledgements}
The authors thank Stefan Milius for inspiring discussions on well-founded
coalgebras, and Dominik Kirst for pointing out the relevance of \Konig's lemma
and its relatives in reverse constructive mathematics.

\bibliographystyle{plainurl}
\bibliography{refs}

\end{document}

%% file: preamble.tex
\DeclareMathOperator{\Sub}{\mathrm{Sub}}

\providecommand{\catname}{\mathsf} 
\providecommand{\clsname}{\mathcal}
\providecommand{\oname}[1]{{\mathop{\mathsf{#1}}\xspace}}

\def\defcatname#1{\expandafter\def\csname B#1\endcsname{\catname{#1}}}
\def\defcatnames#1{\ifx#1\defcatnames\else\defcatname#1\expandafter\defcatnames\fi}
\defcatnames ABCDEFGHIJKLMNOPQRSTUVWXYZ\defcatnames

\def\defclsname#1{\expandafter\def\csname C#1\endcsname{\clsname{#1}}}
\def\defclsnames#1{\ifx#1\defclsnames\else\defclsname#1\expandafter\defclsnames\fi}
\defclsnames ABCDEFGHIJKLMNOPQRSTUVWXYZ\defclsnames

\def\Set{\catname{Set}}

\def\Nom{\catname{Nom}}

\def\Conv{\catname{Conv}}

\newcommand{\Konig}{Kőnig}

\DeclareOldFontCommand{\bf}{\normalfont\bfseries}{\mathbf}

\providecommand{\id}{\mathsf{id}}

\providecommand{\xto}[1]{\,\xrightarrow{#1}\,}

\providecommand{\dar}{\kern-1.2pt\operatorname{\downarrow}}	
\providecommand{\uar}{\kern-1.2pt\operatorname{\uparrow}}	

\providecommand{\inl}{\oname{inl}}
\providecommand{\inr}{\oname{inr}}

\DeclareSymbolFont{Symbols}{OMS}{cmsy}{m}{n}
\DeclareMathSymbol{\iobj}{\mathord}{Symbols}{"3B}
\usepackage{stmaryrd}

\providecommand{\pacman}[1]{}					                     %

\newcommand{\undefine}[1]{\let #1\relax}					                       %

\providecommand{\mone}{{\text{\kern.5pt\rmfamily-}\mathsf{\kern-.5pt1}}}

\makeatletter
\def\mfix#1{\oname{#1}\@ifnextchar\bgroup\@mfix{}}	       %
\def\@mfix#1{#1\@ifnextchar\bgroup\mfix{}}			           %
\makeatother

\providecommand{\case}[3]{\mfix{case}{\mathbin{}#1}{of}{#2}{\kern-1pt;}{\mathbin{}#3}}

\usepackage{dsfont}

\DeclareMathSymbol{\mathinvertedexclamationmark}{\mathord}{operators}{'074}
\DeclareMathSymbol{\mathexclamationmark}{\mathord}{operators}{'041}
\makeatletter
\newcommand{\raisedmathinvertedexclamationmark}{%
  \mathord{\mathpalette\raised@mathinvertedexclamationmark\relax}%
}
\newcommand{\raised@mathinvertedexclamationmark}[2]{%
  \raisebox{\depth}{$\m@th#1\mathinvertedexclamationmark$}%
}
\makeatother

\newcommand{\lfplong}{locally finitely presentable\xspace}
\newcommand{\Lfplong}{Locally finitely presentable\xspace}
\newcommand{\lfp}{\lfplong}
\newcommand{\Lfp}{\Lfplong}

\newcommand{\Pow}{\mathcal{P}}

\DeclareMathOperator{\colim}{\mathsf{colim}}

\newcommand{\initial}{\raisedmathinvertedexclamationmark}

\newcommand{\qqand}{\qquad\text{and}\qquad}

\newcommand{\Powf}{\mathcal{P}_{\omega}}

\newcommand{\ar}{\mathsf{ar}}

\newcommand{\supp}{\mathsf{supp}}
\newcommand{\epito}{\twoheadrightarrow}

\renewcommand{\S}{{\mathcal{S}}}

\newcommand{\seq}{\subseteq}
\newcommand{\ol}{\overline}

\providecommand{\C}{}
\providecommand{\D}{}

\newcommand{\fp}{{\mathrm{fp}}}
\newcommand{\fg}{{\mathrm{fg}}}
\newcommand{\rec}{{\mathrm{rec}}}
\newcommand{\wf}{{\mathrm{wf}}}

\newcommand{\fwf}{\Coalg_{\fp,\wf}(H)}

\renewcommand{\C}{{\mathbb{C}}}
\renewcommand{\D}{{\mathbb{D}}}

\newcommand{\A}{\mathcal{A}}

\newcommand{\iref}[2]{\autoref{#1}.\ref{#1:#2}}

\renewcommand{\id}{{\mathsf{id}}}
\newcommand{\Nat}{\mathds{N}}

\newcommand{\f}{\oname{f}}

\newcommand{\takeout}[1]{\empty}

\newcommand{\Dist}{\mathcal{D}}
\newcommand{\CPowf}{\Pow_{\mathsf{c},\omega}}

\DeclareMathOperator{\Coalg}{\mathsf{Coalg}}
\DeclareMathOperator{\Alg}{\mathsf{Alg}}

\renewcommand{\rho}{\varrho}

\newcommand{\pullbackangle}[2][]{\arrow[phantom,to path={
                     -- ($ (\tikztostart)!1cm!#2:([xshift=8cm]\tikztostart) $)
                        node[anchor=west,pos=0.0,rotate=#2,
                        inner xsep = 0]
                        {\begin{tikzpicture}[minimum
                        height=1mm,baseline=0,#1]
    \draw[-] (0,0) -- (.5em,.5em) -- (0,1em);
                        \end{tikzpicture}}}]{}}

\makeatletter
\newsavebox{\@brx}
\newcommand{\llangle}[1][]{\savebox{\@brx}{\(\m@th{#1\langle}\)}%
  \mathopen{\copy\@brx\kern-0.5\wd\@brx\usebox{\@brx}}}
\newcommand{\rrangle}[1][]{\savebox{\@brx}{\(\m@th{#1\rangle}\)}%
  \mathclose{\copy\@brx\kern-0.5\wd\@brx\usebox{\@brx}}}
\makeatother

\renewcommand{\c}{\colon}

\newcommand{\xra}{\xrightarrow}
\renewcommand{\xto}{\xra}

\newcommand{\monoto}{\rightarrowtail}
\newcommand{\subto}{\hookrightarrow}

\renewcommand{\Nat}{\mathbb{N}}

\newcommand{\gcat}{\mathbb{C}}

\renewcommand{\epsilon}{\varepsilon}

\newcommand*\xbar[1]{%
  \kern.2em\hbox{%
    \vbox{%
      \hrule height 0.5pt %
      \kern0.5ex%
      \hbox{%
        \kern-0.2em%
        \ensuremath{#1}%
        \kern-0.4em%
      }%
    }%
  }\kern.4em %
} 

\makeatletter
\newcommand{\monto}{\@ifstar{\@mtolifted}{\@mto}}
\newcommand{\@mto}{\multimapdot}
\newcommand{\@mtolifted}{\mathbin{\xbar{\multimapdot}}}
\makeatother

\newcommand{\app}[2]{\,}